\documentclass[letterpaper,12pt,peerreviewca,draftcls]{IEEEtran}
\usepackage{csm16}
\usepackage[margin=1in]{geometry}
\usepackage{amsmath} 

\usepackage{url}
\usepackage{graphicx,xcolor}
\usepackage{verbatim}
\makeatletter
\newcommand{\verbatimfont}[1]{\def\verbatim@font{#1}}%
\makeatother
\verbatimfont{\ttfamily\small}

\newcommand{\bi}{\begin{itemize}}\newcommand{\ei}{\end{itemize}}
\newcommand{\be}{\begin{equation}}\newcommand{\ee}{\end{equation}}
\newcommand{\bee}{\begin{enumerate}}\newcommand{\eee}{\end{enumerate}}
\newcommand{\bea}{\begin{eqnarray}}\newcommand{\eea}{\end{eqnarray}}
\newcommand{\beas}{\begin{eqnarray*}}\newcommand{\eeas}{\end{eqnarray*}}
\newcommand{\bc}{\begin{center}}\newcommand{\ec}{\end{center}}
\usepackage[left,pagewise]{lineno} 
\usepackage[english]{babel} 
\usepackage{blindtext}

\usepackage{amsfonts}
\usepackage{amssymb}
\usepackage{amsthm}
\usepackage[compress]{cite} 
\usepackage{caption}
\usepackage{upquote,textcomp}  

\usepackage{etex} 
\usepackage{tikz}
\usetikzlibrary{positioning,plotmarks, matrix, arrows, calc, shapes} 
\usepackage{pgfplots}  

\newtheorem{lemma}{Lemma}

\newtheorem{theorem}{Theorem}
\newcommand{\field}[1]{\mathbb{#1}}
\newcommand{\R}{\field{R}}
\newcommand{\C}{\field{C}}

\theoremstyle{definition}
\newtheorem{defin}{Definition}

\newtheorem{exNoSymbol}{Example}
\newenvironment{ex}
  {\pushQED{\qed}\exNoSymbol}
  {\popQED\endexNoSymbol}

\newcommand{\bmtx}{\begin{bmatrix}}
\newcommand{\emtx}{\end{bmatrix}}
\newcommand{\bsmtx}{\left[ \begin{smallmatrix}} 
\newcommand{\esmtx}{\end{smallmatrix} \right]} 
\newcommand{\bmatarray}[1]{\left[\begin{array}{#1}}
\newcommand{\ematarray}{\end{array}\right]}

\newlength\figurewidth 
\newlength\figureheight

\title{An Introduction to Disk Margins} 

\author{Peter Seiler, Andrew Packard, and Pascal Gahinet \\
  POC: P.\ Seiler (pseiler@umich.edu)\\ \today }

\newif\ifPDF \ifx\pdfoutput\undefined\PDFfalse \else\ifnum\pdfoutput > 0\PDFtrue \else\PDFfalse \fi \fi
\ifPDF 
\usepackage[pdftex, plainpages = false, colorlinks=true, linkcolor=black, citecolor = green!50!blue, urlcolor = blue, filecolor=black, pagebackref=false, hypertexnames=false,  pdfpagelabels ]{hyperref}
\fi

\begin{document}
\maketitle
\CSMsetup

\noindent
Note: Andrew Packard passed away on September 30, 2019 after a long
battle with cancer. He contributed substantially to this work
including the drafting of the paper.

\vspace{0.1in}

Feedback controllers are designed to ensure stability and achieve a
variety of performance objectives including reference tracking and
disturbance rejection. Control engineers have developed different
types of ``safety factors'' to account for the mismatch between the
plant model used for control design and the dynamics of the real
system.  Classical margins account for this mismatch by introducing
gain and phase perturbations in the feedback. The classical margins are
measures of the gain and phase perturbations that can be tolerated
while retaining closed-loop stability.

This paper first reviews classical margins and discusses several
important factors that must be considered with their use. First, real
systems differ from their mathematical models in both magnitude
\textit{and} phase. These simultaneous perturbations are not captured
by the classical margins which only consider gain \textit{or} phase
perturbations but not both.  Second, a small combination of gain and
phase perturbation may cause instability even if the system has large
gain/phase margins. This can be especially important when using
automated computer-based control design over a rich class of
controllers.  The optimization process may improve both gain and phase
margins while degrading robustness with respect to simultaneous
variations.  Third, margin requirements must account for the increase
in model uncertainty at higher frequencies.  All design models lose
fidelity at high frequencies.  Typical gain/phase margin requirements,
e.g. $\pm 6$dB and $45^o$, are sufficient only if the corresponding
critical frequencies remain within the range where the design model is
relatively accurate.  Fourth, there are alternative robustness margins
that provide more useful extensions to multiple-input, multiple-output
(MIMO) systems. One such extension, discussed later in the paper, is
the ``multi-loop'' disk margin which accounts for separate,
independent gain/phase variation in multiple channels.

The paper next introduces disk margins as a tool for
assessing robust stability of feedback systems. Disk margins
address, to some degree, the issues regarding classical margins as
summarized above.  These margins are defined using a general family of
complex perturbations that account for simultaneous gain and phase
variations.  Each set of perturbations, denoted $D(\alpha,\sigma)$, is a
disk parameterized by a size $\alpha$ and skew $\sigma$.  Given a
skew $\sigma$, the disk margin is the largest size $\alpha$ for
which the closed loop remains stable for all perturbations in
$D(\alpha,\sigma)$.  Theorem~\ref{thm:edm} gives an easily computable
expression for the disk margin.  The expression originates from a
variation of the Small Gain Theorem
\cite{zhou96,dullerud00,skogestad05} and provides a construction for
the ``smallest'' destabilizing complex (gain and phase)
perturbation. This complex perturbation can be interpreted as dynamic,
linear time-invariant (LTI) uncertainty.  This is useful as the
destabilizing LTI perturbation can be incorporated within higher
fidelity nonlinear simulations to gain further insight.
Frequency-dependent disk margins can also be computed which provides
additional insight into potential robustness issues.

The class of disk margins defined using $D(\alpha,\sigma)$ includes several
common cases that appear in the literature. First, they include the
symmetric disk margins introduced in \cite{barrett80} and more
recently discussed in \cite{blight94,bates02}. Second, the general
disk margins include conditions based on the distance from the Nyquist
curve of the loop transfer function to the critical $-1$ point
\cite{smith58,franklin18,falcoz15}.  This is related to an
interpretation of disk margins as exclusion regions in the Nyquist
plane. Third, the general disk margins include conditions based on
multiplicative uncertainty models used in robust control
\cite{skogestad05,zhou96}.

Finally, the paper reviews the use of disk margins for MIMO feedback
systems.  A typical extension of classical margins for MIMO systems is
to assess stability with a gain or phase perturbation introduced in a
single channel.  This ``loop-at-a-time'' analysis fails to capture the
effect of simultaneous perturbations occurring in multiple
channels. Disk margins are extended to account for multiple-loop
perturbations.  This multiple-loop analysis provides an introduction
to more general robustness frameworks, e.g.  structured singular value
$\mu$ \cite{doyle78,safonov80,doyle82,doyle85,packard93,fan91} and
integral quadratic constraints \cite{megretski97}.

\section{Background}

This section reviews background material related to dynamical systems
and single-input, single-output (SISO) classical control. This
material can be found in standard textbooks on classical control
\cite{nise10,franklin18,dorf16,ogata09}.

\subsection{Classical Margins}


Consider the classical feedback system shown in
Figure~\ref{fig:gpmfb}.  The plant $P$ and controller $K$ are both
assumed to be linear time-invariant (LTI) and single-input /
single-output (SISO) systems.  The extension to multiple-input /
multiple-output (MIMO) systems is considered later. Assume the
controller $K$ was designed to stabilize the nominal model $P$.
Because this nominal model is only an approximation for the ``real''
dynamics of the plant, control engineers have developed various types
of safety factors to account for the mismatch between the plant model
$P$ and the dynamics of the real system. One way to account for this
mismatch is to introduce the complex-valued perturbation $f$ in
Figure~\ref{fig:gpmfb}. Let $L:=PK$ denote the nominal loop transfer
function.  The perturbed open-loop response is
$L_f := f L$ and the nominal design corresponds to $f = 1$. As $f$
moves away from 1, the closed-loop poles can transition from the open
left-half plane (stable) into the closed right half plane
(unstable). The classical gain and phase margins measure how far $f$
can deviate from $f=1$ while retaining closed-loop stability.


\begin{figure}[h]
\centering
\tikzstyle{block} = [draw, rectangle, minimum width=1.2cm, 
     minimum height=1.2cm, align=center]
\tikzstyle{sum} = [draw,circle,inner sep=0mm,minimum size=3mm] 
\begin{tikzpicture}[auto,>=stealth',very thick,node distance = 1.5cm]
\node [coordinate, name=input] {};
\node[sum, right = of input](sum0){};
\node[block, right=of sum0](Controller) {$K$};
\node[block, right=of Controller](Gain) {$f$};
\node[block, right=of Gain](Plant) {$P$};
\node[coordinate, right = 0.8cm of Plant] (feedback) {};
\draw[->] (input) -- (sum0.west) node[pos=0.2]{$r$};
\draw[->] (sum0.east) -- (Controller.west)  node[pos=0.4]{$e$};
\draw[->] (Controller.east) -- (Gain.west) node[pos=0.6]{$u$};
\draw[->] (Gain.east) -- (Plant.west);
\draw[->] (Plant.east) --  ++(2cm,0cm) node[pos=0.8]{$y$};
\draw[->] (feedback) --  ++(0cm,-1.7cm) -| (sum0) 
    node[pos=0.8, right] {$-$};
\draw[gray,dashed] (2.9cm,1cm) rectangle (10.3cm,-1cm);
\node at (7cm,1.4cm) {$L_f:=f L$};
\end{tikzpicture} 
\caption{Feedback system including perturbation $f$.}
\label{fig:gpmfb}
\end{figure}
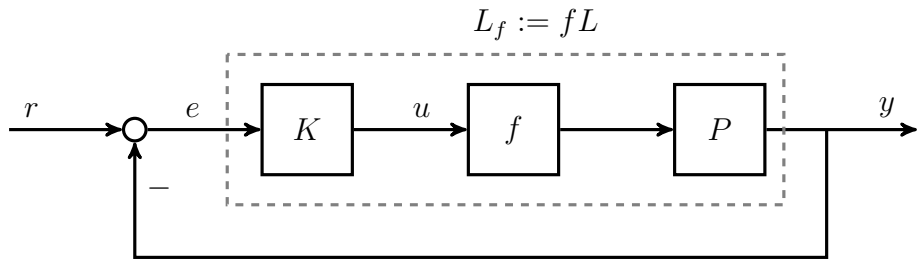

The gain margin measures the amount of allowable perturbation in the
plant gain. This corresponds to real perturbations $f:=g \in \R$.
In other words, the model used for design is $P$ but the real dynamics
might have a different gain as represented by $gP$. It is typically
assumed that the gain of the design model at least has the correct
sign and hence only positive variations $g> 0$ are of interest.  The
gain margin specifies the minimum and maximum variation for which the
closed loop remains stable and well-posed as defined below.

\begin{defin}
  The \textit{gain margins} consist of an upper limit $g_U > 1$ and a
  lower limit $g_L < 1$ such that:
  \begin{enumerate}
  \item the closed-loop is stable and well-posed for all positive gain
    variations $g$ in the range $g_L < g < g_U$,
  \item the closed-loop is unstable or ill-posed for gain variations
    $g = g_U$ (if $g_U<\infty$) and $g = g_L$ (if $g_L>0$).
  \end{enumerate}
  The upper gain margin is $g_U = +\infty$ if the closed-loop remains
  stable and well-posed for all gains $g>1$.  Similarly, the lower
  gain margin is $g_L = 0$ if the closed-loop remains stable and
  well-posed for all positive gains $g < 1$. Reported gain margins are
  often converted to units of decibels, i.e. $20\log_{10}(g)$ where
  $g$ is in actual units.
\end{defin}

The phase margin is the amount of allowable variation in the plant
phase before the closed-loop becomes unstable.  This corresponds to
phase perturbations $f := e^{-j\phi}$ with $\phi \in \R$.  The nominal
loop transfer function is given by $\phi=0$ and $f=1$.  The term phase
variation arises because
$\angle L_{f}(j\omega) = \angle L(j\omega) - \phi$, i.e.  $\phi$
modifies the angle (phase) of the dynamics.  Phase variations can
occur due to time delays in the feedback loop, e.g. due to
implementation on embedded processors, or simply due to deviations in
the plant dynamics.  Sufficient phase margin is required to ensure
that such delays and model variations do not destabilize the system.
It can be shown that the positive and negative phases are equivalent
in a certain sense: $\phi>0$ causes instability if and only if $-\phi$
causes instability. Specifically, the perturbed sensitivity is
$S_\phi(s)=\frac{1}{1+e^{-j\phi}L(s)}$. If $1+fL(j\omega)=0$ for some
perturbation $f=e^{-\phi}$ and frequency $\omega$ then $f$
destabilizes the loop.  Take the complex conjugate of
$1+fL(j\omega)=0$ to show
$1 + \bar{f} \, \overline{L(j\omega)} = 1 + e^{j\phi} L(-j\omega) =
0$.  This implies that $\bar{f} = e^{j\phi}$ also destabilizes the
loop since $1+\bar{f} L(s)$ has a zero at $s=-j\omega$.  

The phase margin specifies the maximum (positive or negative)
variation for which the closed-loop remains stable and well-posed as
defined below.  A related time delay margin can also be defined.


\begin{defin}
  The \textit{phase margin} consists of an upper limit
  $\phi_U \ge 0$ such that:
  \begin{enumerate}
  \item the closed-loop is stable and well-posed for all phase
    variations $\phi$ in the range $-\phi_U < \phi < \phi_U$,
    and
  \item the closed-loop is unstable or ill-posed for
    $\phi = \phi_U$ (if $\phi_U<\infty$).
  \end{enumerate}
  The phase margin is $\phi_U = +\infty$ if the closed-loop remains
  stable and well-posed for all phases $\phi_U>0$. Reported phase
  margins are often converted to units of degrees, i.e.
  $\phi \times \frac{180^o}{\pi}$ where $\phi$ is in radians.  (Note
  that complex numbers repeat with every $360^o=2\pi$ change in
  phase, i.e.  $e^{j\phi} = e^{j\phi+2\pi}$.  The phase margin
  $\phi_U=180^o$ indicates the closed-loop is stable/well-posed for
  $-180^o < \phi < +180^o$ but unstable or ill-posed for
  $\phi=180^o$.  The convention $\phi_U=+\infty$ is equivalent to
  stability for all phases in the range $-180^o$
  $\le \phi \le +180^o$.)
\end{defin}


There is a simple necessary and sufficient condition to compute gain
and phase margins. The nominal closed-loop is assumed to be stable and
hence the poles are in the LHP.  The poles may transition from the LHP
(stable) to the RHP (unstable) due to the gain or phase variation.
The smallest variation that causes the transition from stable to
unstable occurs when a closed-loop pole crosses the imaginary
axis. This occurs when a gain or phase variation places a closed-loop
pole on the imaginary axis at $s=j\omega_0$.  The condition for this
stability transition is: \textit{a gain $f_0=g_0$ or phase
  $f_0=e^{-j\phi_0}$ places a closed-loop pole on the imaginary axis
  at $s=j\omega_0$ if and only if $1+f_0 L(j\omega_0) = 0$.}  This
condition causes the perturbed closed-loop sensitivity
$S_{f_0}:=\frac{1}{1+f_0L}$ to have a pole at $s=j\omega_0$. The gain
margin is the smallest factor $g$ (relative to $g=1$) that puts a
closed-loop pole on the imaginary axis, and similarly for the phase
margin.  This condition can be used to compute gain and/or phase
margins from the Bode plot of the nominal loop $L$. It also suggests a
bisection method to numerically compute the gain and phase
margins. An example is provided next as a brief review of the
classical margins.

\begin{ex}
  \label{ex:cm}


  Consider a feedback system with the following plant $P$, controller $K$,
  ad nominal loop $L$:
  \begin{align}
    P(s) = \frac{1}{s^3+10s^2+10s+10}, \,\,\, 
    K(s)=25, \,\,\,
    L(s) = \frac{25}{s^3+10s^2+10s+10}.
  \end{align}
  The nominal closed-loop has poles in the LHP at $-9.33$ and
  $-0.33 \pm 1.91j$ and hence is stable.  The poles of the closed-loop
  system remain in the LHP for all gain variations $f=g<1$. Hence the
  lower gain margin is $g_L=0$.  However, the closed-loop poles cross
  into the RHP for sufficiently large gains $g>1$. The upper gain
  margin $g_U=3.6$ marks the transition as poles move from the LHP
  (stable) into the RHP (unstable).  The closed-loop is stable for
  $g\in[0,g_U)$. For $g=g_U$ the closed-loop has poles on the
  imaginary axis $s=\pm j \omega_1$ at the critical frequency
  $\omega_1 = 3.16$ rad/sec.  In other words, $1+g_UL(j\omega_1)=0$
  and it can be verified that the perturbed sensitivity
  $S_{g_U} = \frac{1}{1+g_U L}$ has a poles on the imaginary axis at
  $s=\pm j \omega_1$.  The poles of the closed-loop also cross into
  the RHP as the phase increases. The phase margin $\phi_U=29.1^o$
  marks the transition as poles move from the LHP (stable) into the
  RHP (unstable).  The closed-loop is stable for
  $\phi\in (-\phi_U,\phi_U)$. For $\phi=\phi_U$ the closed-loop has
  poles on the imaginary axis $s=\pm j \omega_2$ at the critical
  frequency $\omega_2 = 1.78$ rad/sec. Again, this corresponds to
  $1+e^{-j\phi_U} L(j\omega_2)=0$ and it can be verified that the
  perturbed sensitivity has a poles on the imaginary axis at
  $s=\pm j \omega_2$.
\end{ex}








\subsection{Limitations of Classical Margins}

There are several important factors that must be considered when
using classical margins:
\begin{enumerate}

\item \textit{Real systems differ from their mathematical models in
    both magnitude \underline{and} phase:} The Bode plot in
  Figure~\ref{fig:GMPMLimitations} shows a collection of frequency
  responses obtained from input-output experiments on hard disk drives
  (blue).  A low order model used for control design is also shown
  (yellow). The model accurately represents the experimental data up
  to 2-3rad/sec but the experimental data has both gain and phase
  variations at higher frequencies. These simultaneous perturbations
  are not captured by the classical margins which only consider gain
  \textit{or} phase perturbations but not both.


\begin{figure}[h!]
  \centering
  \includegraphics[width=0.57\textwidth]{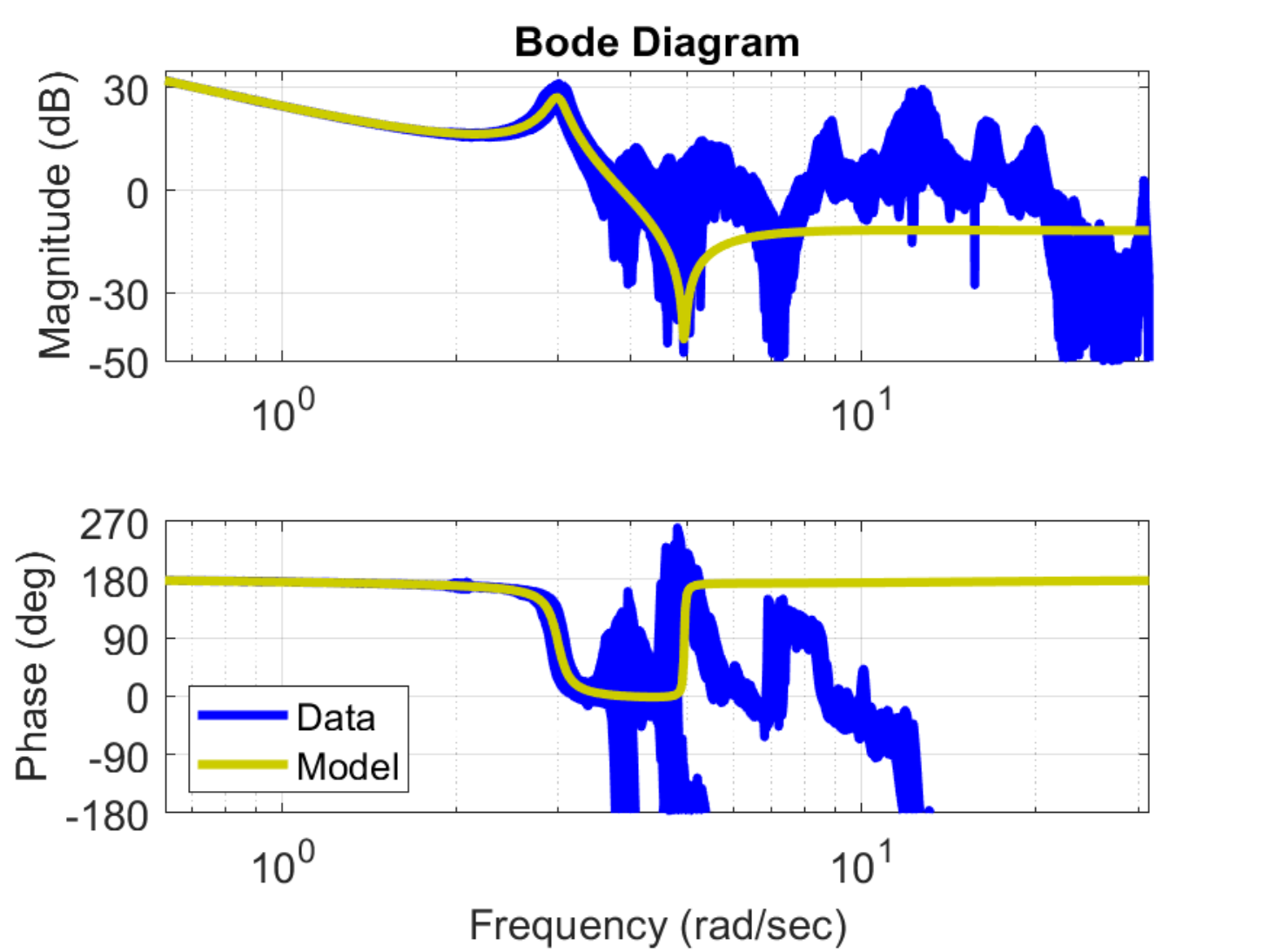}
  \caption{Experimental frequency responses from many hard disk
    drives (blue) and a low-order design model (yellow). 
    This data is provided by Seagate and the frequency axis has been
    normalized for proprietary reasons.}
  \label{fig:GMPMLimitations}
\end{figure}

\item \textit{Small plant perturbations may cause robustness issues
    even if the system has large gain/phase margins:} Real systems
  have simultaneous gain and phase perturbations as noted in the first
  comment. Moreover, there are examples of systems with large gain and
  phase margins but for which a small (combined) gain/phase
  perturbation causes instability. See Section 9.5 of \cite{zhou96}
  for the construction of such an example. An extreme example is given
  by the following loop:

{ \small
  \begin{align}
    \label{eq:badL}
    L(s):=\frac{-47.252 s^7 - 20.234 s^6 - 135.4086 s^5 + 61.6166 s^4 
    + 804.6454 s^3 + 600.0611 s^2 + 59.1451 s + 1.888}{99.8696 s^7 
    + 175.5045 s^6 + 673.7378 s^5 + 890.5109 s^4 + 553.1742 s^3 
    - 49.2268 s^2 + 12.1448 s + 1}.
  \end{align}
} 

Figure~\ref{fig:BadMarginsEx} shows a portion of the Nyquist plot for
this loop. The feedback system with $L$ has phase margin
$\phi_U = 45^o$ and gain margins $[g_L,g_U]=[0.2,2.1]$. The points
corresponding to the phase margin and upper gain margin
($-1/g_U$) are marked with green squares in the figure. The
classical margins are large but the Nyquist curve for $L$ comes near
to the $-1$ point.  Thus small (simultaneous gain and phase)
perturbations can cause the feedback system to become unstable.  The
key point is that some care is required when using classical gain and
phase margins.  This did not present itself as an issue when
controllers were designed primarily with graphical techniques.  These
classical controllers were typically of limited complexity and did not
have enough degrees of freedom to get into this corner. However, this
issue can be especially important when using automated computer-based
control design over a rich class of controllers.  The optimization
process may improve both gain and phase margins while degrading
robustness with respect to simultaneous variations.


 \begin{figure}[h!]
   \centering
   \includegraphics[width=0.5\textwidth]{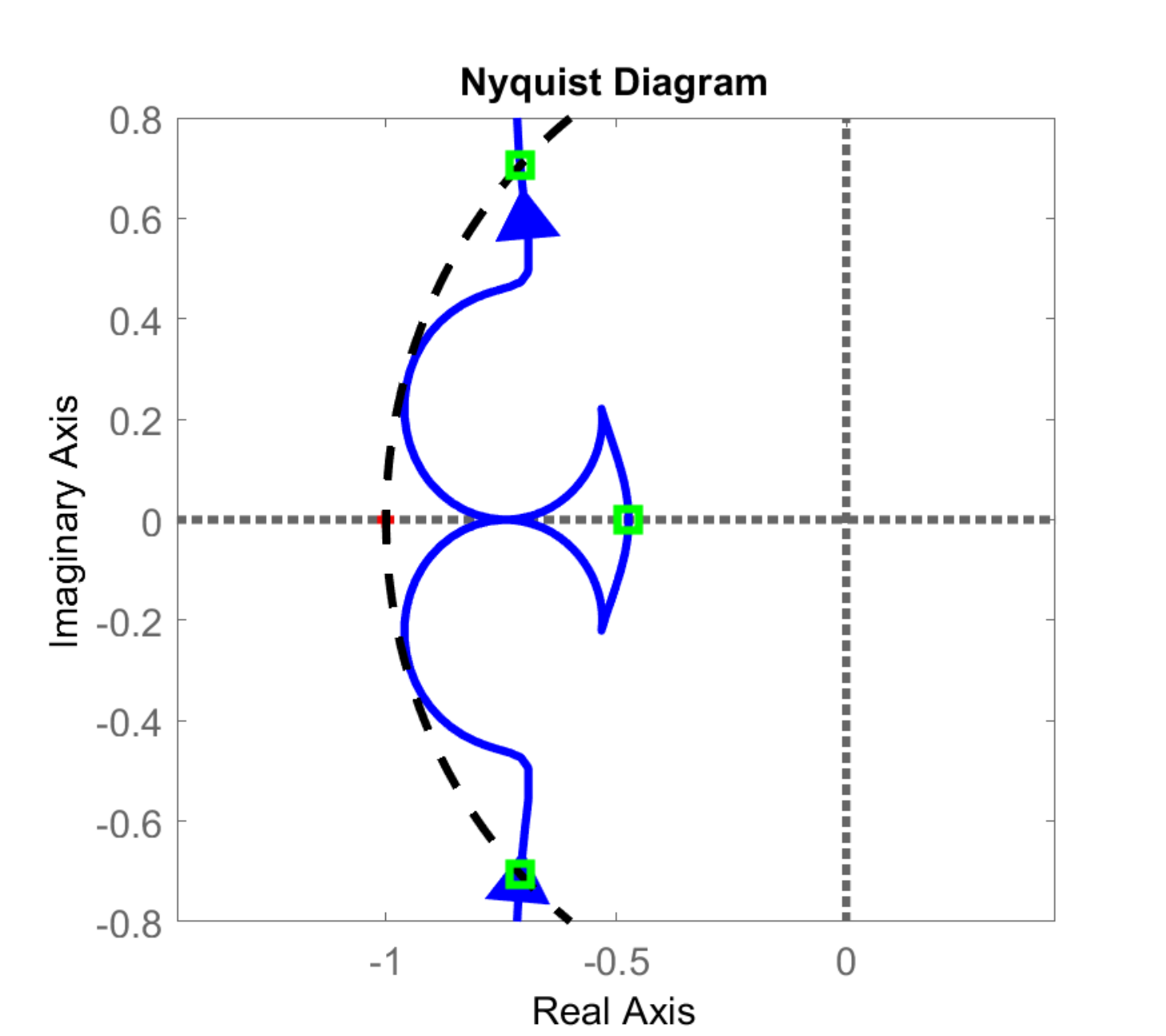}
   \caption{Nyquist plot of the loop $L$ in
     Equation~\ref{eq:badL}. This loop has large classical gain and
     phase margins (green squares) but poor robustness to simultaneous
     gain and phase perturbations.}
   \label{fig:BadMarginsEx}
 \end{figure}

\item \textit{Margin requirements must account for the increase in
    model uncertainty at higher frequencies:} Consider again the hard
  disk drive frequency responses shown in
  Figure~\ref{fig:GMPMLimitations}.  The design model (yellow) loses
  fidelity at high frequencies. As a result, the margins must
  necessarily be larger at higher frequencies to ensure stability.
  Requirements based on simple rules of thumb, e.g. $45^o$ of phase
  margin, are insufficient and must account for the expected level of
  model uncertainty.  For example, the design model for the hard disk
  drive data is relatively accurate at low frequencies.  The typical
  $45^o$ phase margin requirement might be sufficient \textit{if}
  the closed-loop bandwidth remains below 2-3 rad/sec where the design
  model has small perturbations. However, this typical phase margin
  requirement will be insufficient if the closed-loop bandwidth is
  pushed beyond 2-3 rad/sec. 



\item \textit{There are alternative robustness margins that provide
    more useful extensions to MIMO systems:} A typical extension of
  classical margins for MIMO systems is to assess stability with a
  gain or phase perturbation introduced into a single channel. This
  analysis is repeated for each input and output channel.  This
  ``loop-at-a-time'' analysis fails to capture the effect of
  simultaneous perturbations occurring in multiple channels.  Hence it
  can provide an overly optimistic view of robustness. Alternative
  robustness margins are more easily extended to account for
  ``multiple-loop'' perturbations as discussed later in the paper.

\end{enumerate}

\section{SISO Disk Margins}

This section introduces the notion of disk margins for SISO systems as
a tool to address some of the limitations of classical margins.  Disk
margins are robust stability measures that account for simultaneous
gain and phase perturbations.  They also provide additional
information regarding the impact of model uncertainty at various
frequencies.

\subsection{Modeling Gain and Phase Variations}

Gain and phase variations are naturally modeled as a complex-valued
multiplicative factor $f$ acting on the open-loop $L$
yielding a perturbed loop $L_f= f L$.  This factor is nominally
1 and its maximum deviation from $f=1$ quantifies the amount of gain
and phase variation. A family of such models is given by:
\begin{align}
\label{eq:Fab}
f \in  D(\alpha,a,b) = \left\{ \frac{1+a \delta}{1-  b \delta} \, : \, 
\delta \in \C \mbox{ with } |\delta| < \alpha \right\},
\end{align}
where $a,b,\alpha$ are real parameters that define the set of
perturbations. The sets $D(\alpha,a,b)$ contain $f=1$, corresponding
to $\delta=0$, and are delimited by a circle centered on the real axis
(assuming $|b\alpha|<1$). For example, the set $D(\alpha,a,b)$ for
$a=0.4$, $b=0.6$ and $\alpha = 0.75$ is the shaded disk shown in
Figure~\ref{fig:FSet}. Note that the nominal value $f=1$ is not
necessarily at the disk center $c$. The real axis intercepts
$\gamma_{\max}$ and $\gamma_{\min}$ determine the maximum relative
increase and decrease of the gain.  The line from the origin and
tangent to the disk determines the maximum phase variation
$\phi_{\max}$ achieved by any perturbation $f \in D(\alpha,a,b)$.

\begin{figure}[h!]
  \centering
  \includegraphics[width=0.55\textwidth]{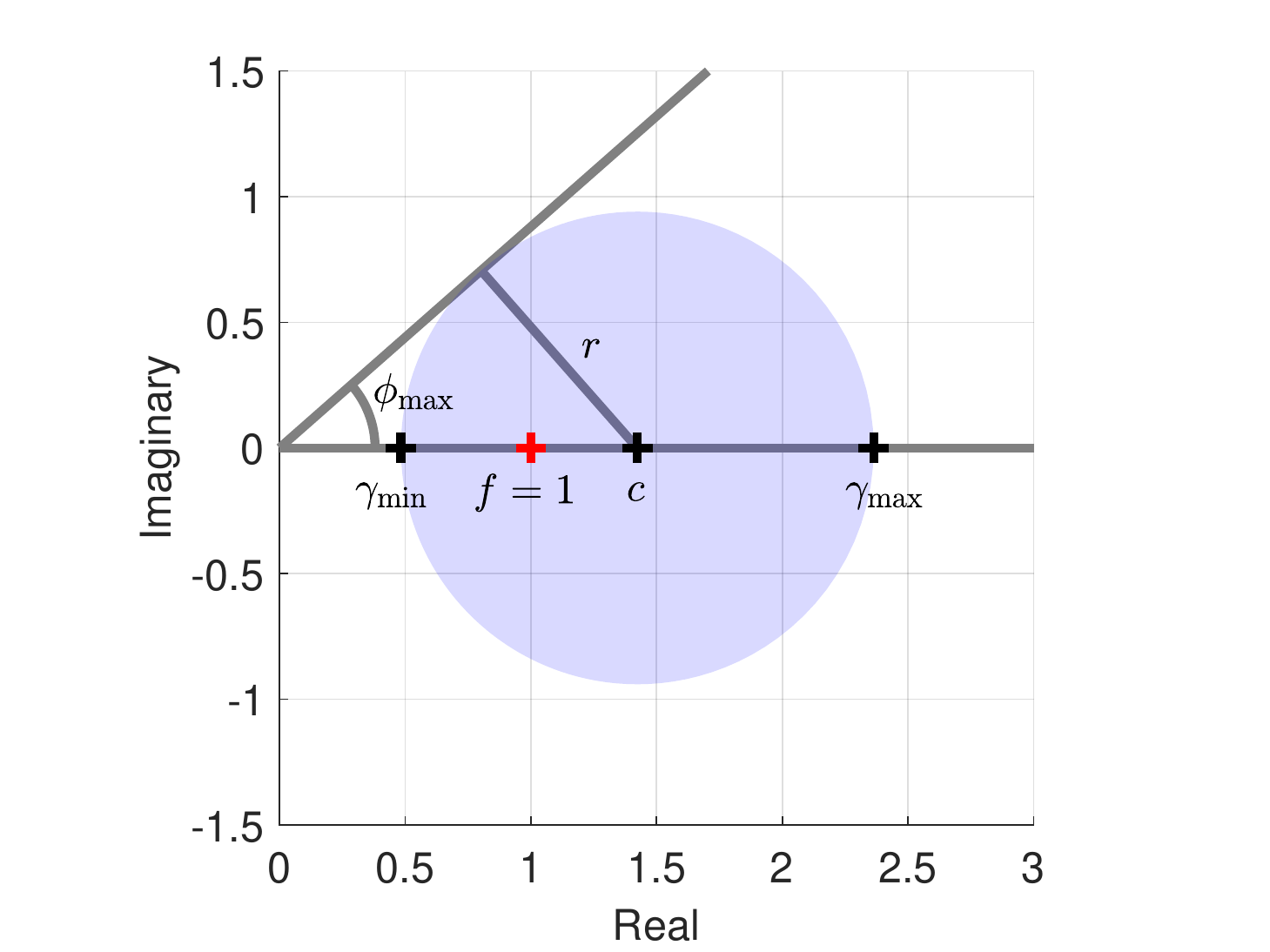}
  \caption{Set of variations $D(\alpha,a,b)$ for $a = 0.4$,
    $b = 0.6$, and $\alpha = 0.75$. This is equivalent to
    $D(\alpha,\sigma)$ for $\sigma=0.2$ and $\alpha=0.75$.}
  \label{fig:FSet}
\end{figure}

There are two issues with the family of models in
Equation~\ref{eq:Fab}.  First, if $a=-b$ then the set only
contains the point $f=1$.  Thus $a+b\ne 0$ is required to avoid this
degenerate case. Second, the set is unchanged when multiplying $a,b$
by some constant and dividing $\alpha$ by the same constant, i.e.
$D\left(\frac{\alpha}{|\kappa|},\kappa a, \kappa b \right) =
D(\alpha,a,b)$ for any $\kappa\ne 0$. This suggests further imposing
$a+b=1$. It is useful to parameterize these constants as
$a := \frac{1}{2}(1-\sigma)$ and $b := \frac{1}{2}(1+\sigma)$ where $\sigma\in\R$ is
a \textit{skew} parameter. This yields the simplified
parameterization:
\begin{align}
\label{eq:Fe}
  f \in  D(\alpha,\sigma) = \left\{ 
  \frac{1+ \frac{1-\sigma}{2} \, \delta }{1- \frac{1+\sigma}{2}
  \, \delta} \, : \, 
  \delta \in \C  \mbox{ with } |\delta| < \alpha \right\}.
\end{align}
Again, the sets $D(\alpha,\sigma)$ are delimited by circles centered on the
real axis (assuming $|\frac{1}{2}(1+\sigma)\alpha|<1$). The disk in
Figure~\ref{fig:FSet} is defined, in this simplified parameterization,
by the choices $\sigma=0.2$ and $\alpha=0.75$.  The intercepts on the real
axis correspond to $\delta=\pm \alpha$ and are given by:
\begin{align}
\label{eq:gminmax}
\gamma_{\min} = \frac{2-\alpha(1-\sigma)}{2+\alpha(1+\sigma)} 
\,  \mbox{ and } \,
\gamma_{\max} = \frac{2+\alpha(1-\sigma)}{2-\alpha(1+\sigma)}.
\end{align}
The disk center and radius are:
\begin{align}
\label{eq:Fcandr}
c = \frac{1}{2} (\gamma_{\min}+\gamma_{\max}) 
\,  \mbox{ and } \,
r = \frac{1}{2} (\gamma_{\max}-\gamma_{\min}). 
\end{align}
The maximum phase variation satisfies $\sin\phi_{max} = \frac{r}{c}$
when $r\le c$.  This follows from the right triangle formed from the
origin, disk center, and point where the tangent line intersects
$D(\alpha,\sigma)$. If $r>c$ then $D(\alpha,\sigma)$ contains the origin
and $\phi_{max}:=+\infty$.


There is some coupling between $\sigma$ and $\alpha$. However, it is
helpful to think of $\alpha$ as controlling the amount of gain and
phase variation while $\sigma$ captures the difference between the
amount of relative gain increase and decrease.  First consider the
case $\sigma=0$. For this choice we have
$\gamma_{\max} = 1/\gamma_{\min}$, i.e.  the maximum gain increase and
decrease are the same in relative terms. We refer to this as the
\emph{balanced} case.  An example of a balanced disk with $\sigma=0$
and $\alpha=\frac{2}{3}$ is shown in both the left and right subplots
of Figure~\ref{fig:eSkew} (blue disk with dashed outline).  The real
axis intercepts $\gamma_{min}=0.5$ and $\gamma_{max}=2$ are balanced
in the sense that they both correspond to changing the gain by a
factor $2$.  The disk moves to the right when increasing $\sigma$ from
the balanced case $\sigma=0$ and adjusting $\alpha$ to keep the radius
constant. This is illustrated in the right subplot of
Figure~\ref{fig:eSkew}.  This means that $\sigma>0$ models a gain
variation that can increase by a larger factor than it can decrease.
Similarly, decreasing $\sigma$ from the balanced case $\sigma=0$ moves
the disk to the left as shown in the left subplot of
Figure~\ref{fig:eSkew}.  This means that the gain can decrease by a
larger factor than it can increase and that it can even change sign.
For $\sigma=-1$, the disk intercepts are $\gamma_{\min} = 1-\alpha$
and $\gamma_{\max} = 1+\alpha$, i.e., the gain can increase or
decrease by the same absolute amount. These examples clarify the
meaning of the term \emph{skew} for the parameter $\sigma$. For
$\sigma=0$, the nominal factor $f=1$ is the geometric mean of the
range $(\gamma_{\min},\gamma_{\max})$ and it moves off-center when
selecting a positive or negative value for $\sigma$.  In summary, a
skew $\sigma=0$ means that the gain can increase or decrease by the
same factor, i.e. it has a symmetric range of variation in dB. A
nonzero skew indicates a bias, on a logarithmic/dB scale, toward gain
decrease ($\sigma<0$) or gain increase ($\sigma>0$).

\begin{figure}[h!]
  \centering
  \includegraphics[width=0.45\textwidth]{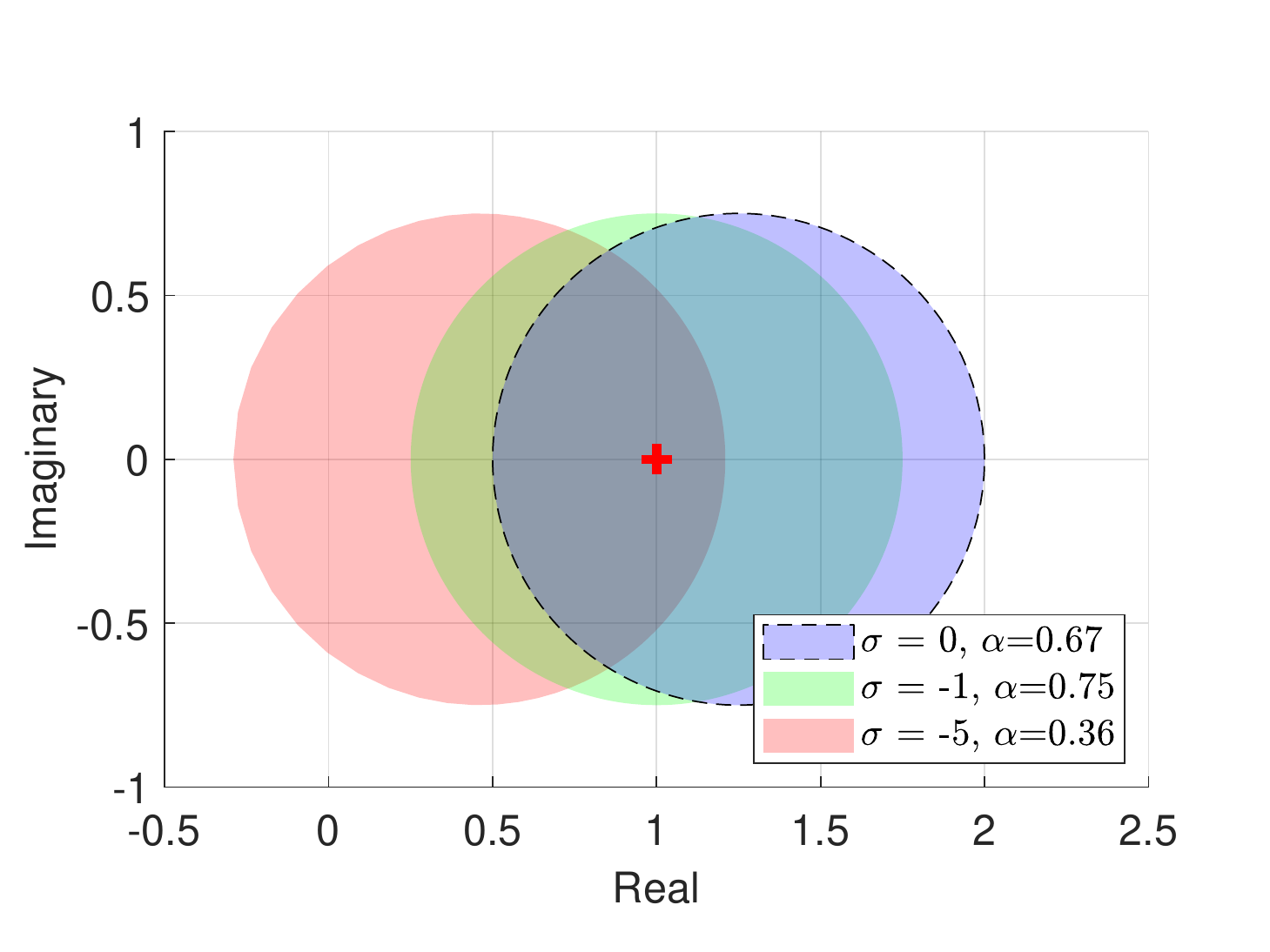}
  \includegraphics[width=0.45\textwidth]{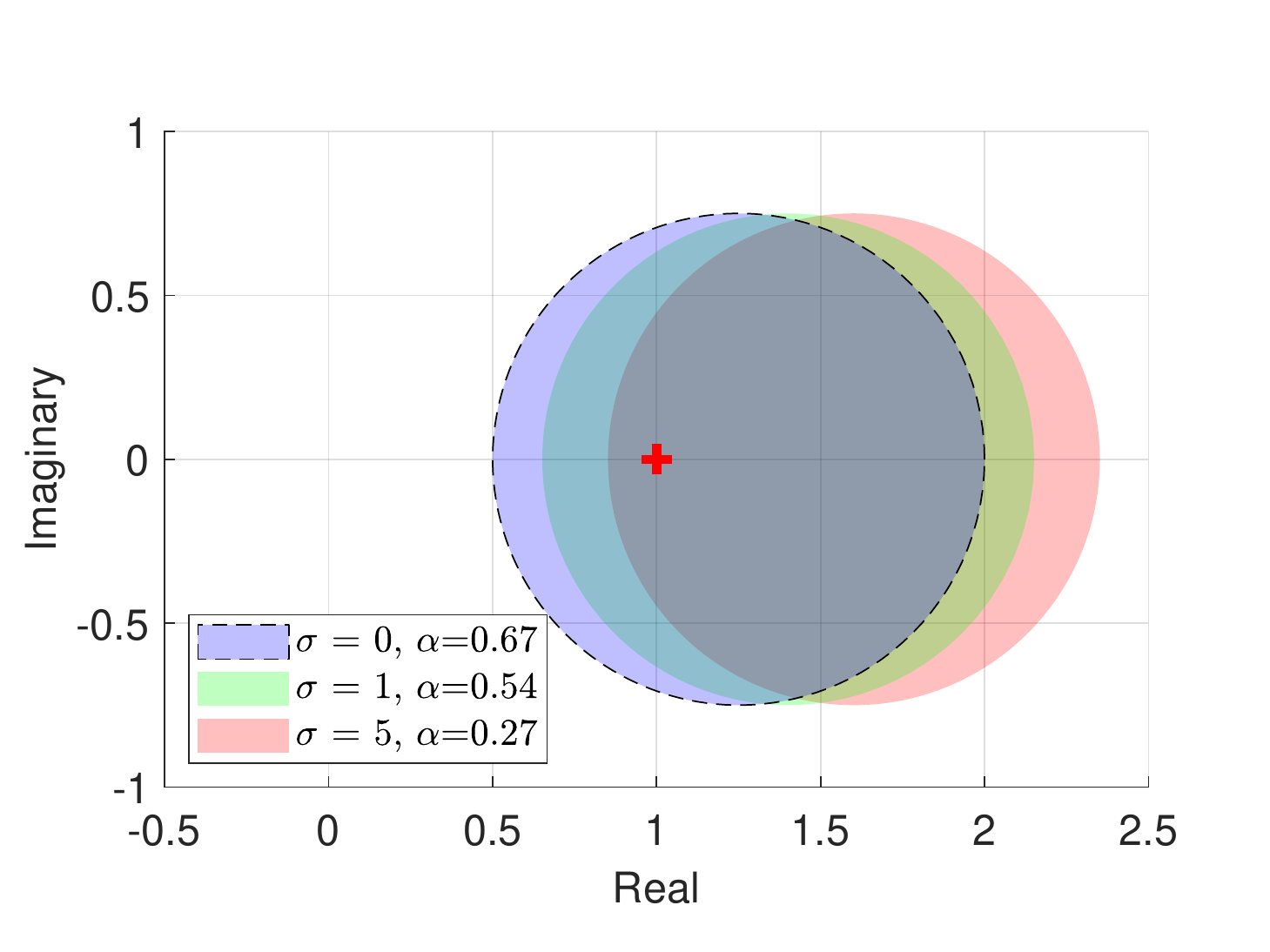}
  \caption{Positive $\sigma$ skews the gain variation right toward more
    gain increase (right).  Negative $\sigma$ skews the gain variation left
    toward more gain decrease (left). The parameter $\alpha$ is
    selected to maintain the same radius for all disks.}
  \label{fig:eSkew}
\end{figure}

For fixed $\sigma$, the parameter $\alpha>0$ controls the size of the
region $D(\alpha,\sigma)$. This is illustrated in
Figure~\ref{fig:alphaSets} for $\sigma=0$. The region is the interior of a
disk for $\alpha < \frac{2}{|1+\sigma|}$.  The size of the disk increases
for larger values of $\alpha$.  The region becomes a half-plane for
$\alpha = \frac{2}{|1+\sigma|}$ and the exterior of a disk for
$\alpha > \frac{2}{|1+\sigma|}$. It can be shown with some algebra that
$\gamma_{max} - \gamma_{min} = 8\alpha / ( 4-\alpha^2(1+\sigma)^2 )$. Thus
if $\alpha > \frac{2}{|1+\sigma|}$ then $\gamma_{max}<\gamma_{min}$,
i.e. $\gamma_{max}$ becomes the ``left'' intercept on the disk.
Equation~\ref{eq:Fcandr} still provides a valid definition for the
disk center $c$ but the disk radius in this less common case is
$r=\frac{1}{2}|\gamma_{max} - \gamma_{min}|$) The case
$\alpha < \frac{2}{|1+\sigma|}$ is most relevant in practice since it
corresponds to the interior of a disk with bounded gain and phase
variations. However, the case $\alpha \geq \frac{2}{|1+\sigma|}$ can be
used to model situations where the gain can vary substantially or the
phase is essentially unknown.  This qualitative analysis provides
guidance on the effect of the parameters $\sigma$ and $\alpha$.


\begin{figure}[h!]
  \centering
  \includegraphics[trim={2.25cm 0 2.25cm 0},clip,width=0.95\textwidth]{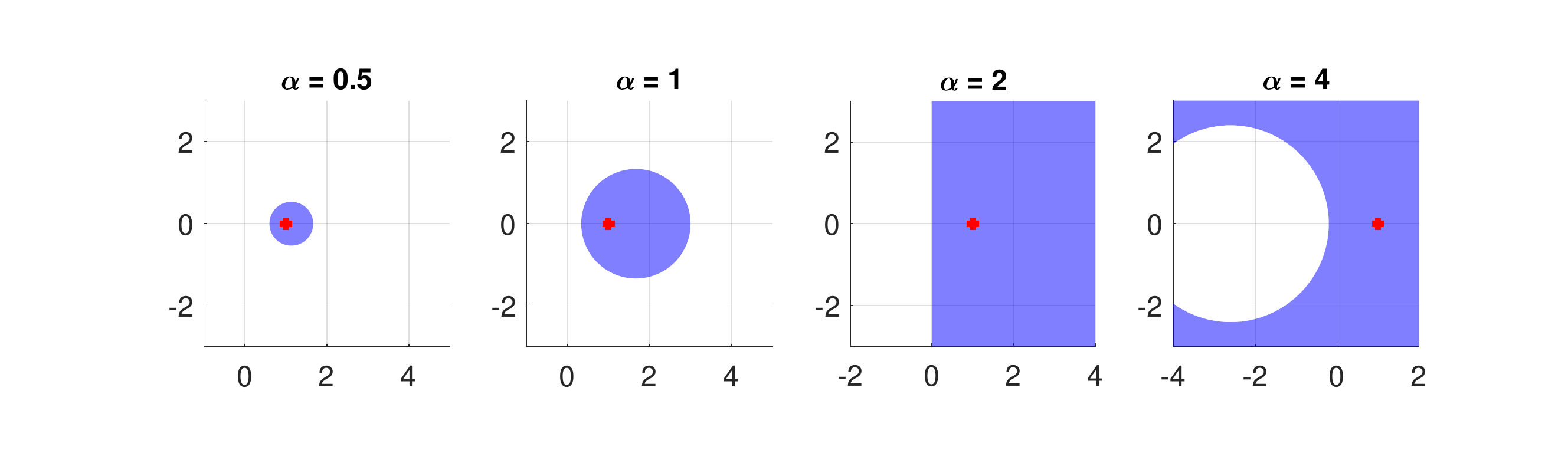}
  \caption{Increasing $\alpha$ increases the size of $D(\alpha,\sigma)$ as
    shown with $\sigma=0$.}
  \label{fig:alphaSets}
\end{figure}

\subsection{Disk Margins: Definition and Computation}

There are two common robustness analyses that can be performed with
the set $D(\alpha,\sigma)$ of gain and phase variations.  The first
approach is to select $\sigma$ and compute the largest value of $\alpha$
for which closed-loop stability is maintained. This yields a stability
margin, formally defined next, that can be used to estimate the
degree of robustness for a feedback loop.
\begin{defin}
  \label{def:Edm}
  For a given skew $\sigma$, the \emph{disk margin} $\alpha_{max}$
  is the largest value of $\alpha$ such that closed-loop with $f L$ is
  well-posed and stable for all complex perturbations $f\in D(\alpha,\sigma)$.
\end{defin}
The set $D(\alpha_{max},\sigma)$ is a stable region for gain and phase
variations, i.e., variations by a factor $f$ inside
$D(\alpha_{max},\sigma)$ cannot destabilize the feedback loop. Note that
the set $D(\alpha,\sigma)$ is not necessarily a disk, as demonstrated in
Figure~\ref{fig:alphaSets}. Hence the term ``disk'', strictly
speaking, refers to the disk $|\delta|<\alpha$.  If little is known
about the distribution of gain variations then $\sigma=0$ is a reasonable
choice as it allows for a gain increase or decrease by the same
relative amount. The choice $\sigma<0$ is justified if the gain can
decrease by a larger factor than it can increase. Similarly, the
choice $\sigma>0$ is justified when the gain can increase by a larger
factor than it can decrease.

An alternative approach is to use $D(\alpha,\sigma)$ to cover known gain
and phase variations, e.g.  neglected actuator or sensor dynamics.
This approach requires some knowledge of the plant modeling errors
specified in terms of gain and phase variations.  Then $\alpha$ and
$\sigma$ are selected to give the smallest set $D(\alpha,\sigma)$ that covers
these known variations.  The goal is then to assess the robustness of
the closed-loop with respect to this set of variations. This second
analysis approach can be performed by computing the disk margin
$\alpha_{max}$ associated with the chosen skew $\sigma$.  If
$\alpha_{max} \ge \alpha$ then the closed-loop is stable for all
variations in $D(\alpha,\sigma)$ and hence the system is robust to the
known modeling errors.

There is a simple expression for the disk margin $\alpha_{max}$. As
with the classical margins, the nominal feedback system is assumed to
be stable and hence the closed-loop poles are in the LHP for
$f=1$. The poles move continuously in the complex plane as
$f \in D(\alpha,\sigma)$ is perturbed away from $f=1$. The poles may move
into the RHP (unstable closed-loop) if $f$ is varied by a sufficiently
large amount from the nominal value $f=1$. The transition from stable
to unstable occurs when the closed-loop poles cross the imaginary
axis.  The condition for this stability transition is: \textit{a
  perturbation $f_0 \in D(\alpha,\sigma)$ places a closed-loop pole on the
  imaginary axis at $s=j\omega_0$ if and only if
  $1+f_0 L(j\omega_0) = 0$.}  (The perturbation $f_0$ is complex and
hence the roots of $1+f_0 L(j\omega_0)=0$ are not necessarily complex
conjugate pairs. However, the disk $D(\alpha,\sigma)$ has conjugate
symmetry.  As a result, if $f_0 \in D(\alpha,\sigma)$ causes a pole
at $s=j\omega_0$ then $\bar f_0 \in D(\alpha,\sigma)$ causes a pole
at $s=-j\omega_0$.)

The definition of $D(\alpha,\sigma)$ (Equation~\ref{eq:Fe}) implies that
$f_0=\frac{2+ (1-\sigma)\delta_0 }{2- (1+\sigma)\delta_0}$ for some
$\delta_0\in \C$ with $|\delta_0|<\alpha$. Thus the stability
transition condition can be re-written, after some algebra, in terms
of the sensitivity $S:=\frac{1}{1+L}$ as follows:
\begin{align}
  \label{eq:SdmCondition}
  \left( S(j\omega_0) + \frac{\sigma-1}{2} \right) \delta_0 = 1.
\end{align}
To summarize, some $f_0 \in D(\alpha,\sigma)$ causes a closed-loop pole at
$s=j\omega_0$ if and only if
$\left( S(j\omega_0) + \frac{\sigma-1}{2} \right) \delta_0 = 1$ holds for
some $|\delta_0| < \alpha$. This condition forms the basis for the
next theorem regarding the disk margin.  The theorem uses the
following notation for the peak (largest value) gain of a stable, SISO,
LTI system $G$:
\begin{align}
  \label{eq:HinfSISO}
  \|G\|_\infty := \max_{\omega \in \R \cup \{+\infty\}}  |G(j\omega)|.
\end{align}
This is called the $H_\infty$ norm for the stable system $G$ and it
corresponds to the largest gain on the Bode magnitude plot.
\begin{theorem}
\label{thm:edm}
  Let $\sigma$ be a given skew parameter defining the disk margin.
  Assume the closed-loop is well-posed and stable with the nominal,
  SISO loop $L$. Then the disk margin is given by:
\begin{align}
  \label{eq:alphadm}
  \alpha_{max} = \frac{1}{\left\| S + \frac{\sigma-1}{2} \right\|_\infty}.
\end{align}
\end{theorem}
\begin{proof}
  A formal proof is given in the appendix entitled''Proof of Disk
  Margin Condition''.  Briefly, consider any
  $f_0\in D(\alpha_{max},\sigma)$ with corresponding
  $|\delta_0|<\alpha_{max}$.  Equation~\ref{eq:alphadm} implies the
  the inequality:
  $\left| S(j\omega) + \frac{\sigma-1}{2} \right| \cdot |\delta_0| < 1$ for
  all $\omega$.  This further implies that
  $\left( S(j\omega) + \frac{\sigma-1}{2} \right) \delta_0 \ne 1$ and
  hence, based on the discussion above, the poles cannot lie on the
  imaginary axis.  Finally, the poles are in the LHP for the nominal
  value $f=1$ and, as just shown, cannot cross the imaginary axis for
  any $f_0 \in D(\alpha_{max},\sigma)$.  Therefore the closed-loop remains
  stable for all $f_0 \in D(\alpha_{max},\sigma)$. The formal proof also
  shows that there is a perturbation $f_0$ on the boundary of
  $D(\alpha_{max},\sigma)$ that causes instability. Hence $\alpha_{max}$
  given in Equation~\ref{eq:alphadm} defines the largest possible
  stable region.
\end{proof}

The margin $\alpha_{max}$ decreases as $\|S+\frac{\sigma-1}{2}\|_\infty$
increases, i.e.  large peak gains of $S+\frac{\sigma-1}{2}$ correspond to
small robustness margins.  Several special cases are often considered
in the literature.  The disk margin condition for the balanced case
($\sigma=0$) can be expressed as
$\alpha_{max} = \| \frac{1}{2}(S-T) \|_{\infty}^{-1}$. This is known
as the \emph{symmetric disk margin} \cite{barrett80,blight94,bates02}
because the disks $D(\alpha_{max},\sigma=0)$ are balanced in terms of the
relative gain increase and decrease.  If $\sigma=-1$ or $\sigma=+1$ then the disk
margin condition simplifies to $\alpha_{max}= \|T\|_{\infty}^{-1}$ and
$\alpha_{max}=\|S\|_{\infty}^{-1}$, respectively.  These special cases
are called $T$-based and $S$-based disk margins.

Efficient algorithms are available to compute both peak gain of an LTI
system and the corresponding peak frequency
\cite{boyd89,bruinsma90}. These can be used to compute
$\|S+\frac{\sigma-1}{2}\|_\infty$ and thus the disk margin.  The formal
proof of Theorem~\ref{thm:edm} also provides an explicit construction
for a destabilizing perturbation $f_0$ on the boundary of
$D(\alpha_{max},\sigma)$. First, compute the frequency $\omega_0$ where
$S+\frac{\sigma-1}{2}$ achieves its peak gain.  Next, evaluate the
frequency response of $S(j\omega_0)$ and define
$\delta_0:=\left( S(j\omega_0)+\frac{\sigma-1}{2} \right)^{-1}$. The
corresponding perturbation
$f_0=\frac{2+ (1-\sigma)\delta_0 }{2- (1+\sigma)\delta_0}$ causes the
closed-loop to be unstable (if $\omega_0$ finite) with a pole on the
imaginary axis at $s=j\omega_0$ or ill-posed (if $\omega_0=\infty$).
If this construction yields $\delta_0 = \frac{2}{\sigma+1}$ then
$f_0=\infty$. This occurs when $S(j\omega_0)=1$ and
$L(j\omega_0) = 0$. This corresponds to the trivial case where
$D(\alpha_{max},\sigma)$ is a half-space
($\alpha_{max} = \frac{2}{|1+\sigma|}$) and the closed-loop
retains stability for any perturbation in this half space.

\begin{ex}
  \label{ex:edm}
  Consider again the loop $L(s)=\frac{25}{s^3+10s^2+10s+10}$
  introduced previously in Example~\ref{ex:cm}.  The feedback system
  with this loop is nominally stable. By Theorem~\ref{thm:edm}, the
  symmetric disk margin for $\sigma=0$ is given by
  $\alpha_{max} = \|\frac{1}{2} \, (S-T) \|_\infty^{-1}$.  The peak
  gain of $\frac{1}{2}(S-T)$ is 2.18 at the critical frequency
  $\omega_0=1.94$ rad/sec. This yields a symmetric disk margin of
  $\alpha_{max} = 0.46$.  The corresponding symmetric disk
  $D(\alpha_{max},\sigma=0)$ has real axis intercepts at
  $\gamma_{min}=0.63$ and $\gamma_{max} = 1.59$.  The closed-loop is
  stable for all gain and phase perturbations in the interior of this
  disk. However, there is a destabilizing perturbation on the boundary
  of $D(\alpha_{max},\sigma=0)$. The construction above yields
  $\delta_0 = 0.212 - 0.406j$ and the destabilizing perturbation
  $f_0 = 1.128-0.483j$.  The closed-loop with this perturbation is
  unstable with a pole at $s=j\omega_0$. Figure~\ref{fig:SBodemag}
  shows the closed-loop sensitivities for the nominal $f=1$ (blue
  solid) and destabilizing perturbation $f_0$ (red dashed).  The
  perturbed sensitivity has infinite gain at the critical frequency
  $\omega_0$ due to the imaginary axis pole.
\end{ex}

\begin{figure}[h!]
  \centering
  \includegraphics[width=0.5\textwidth]{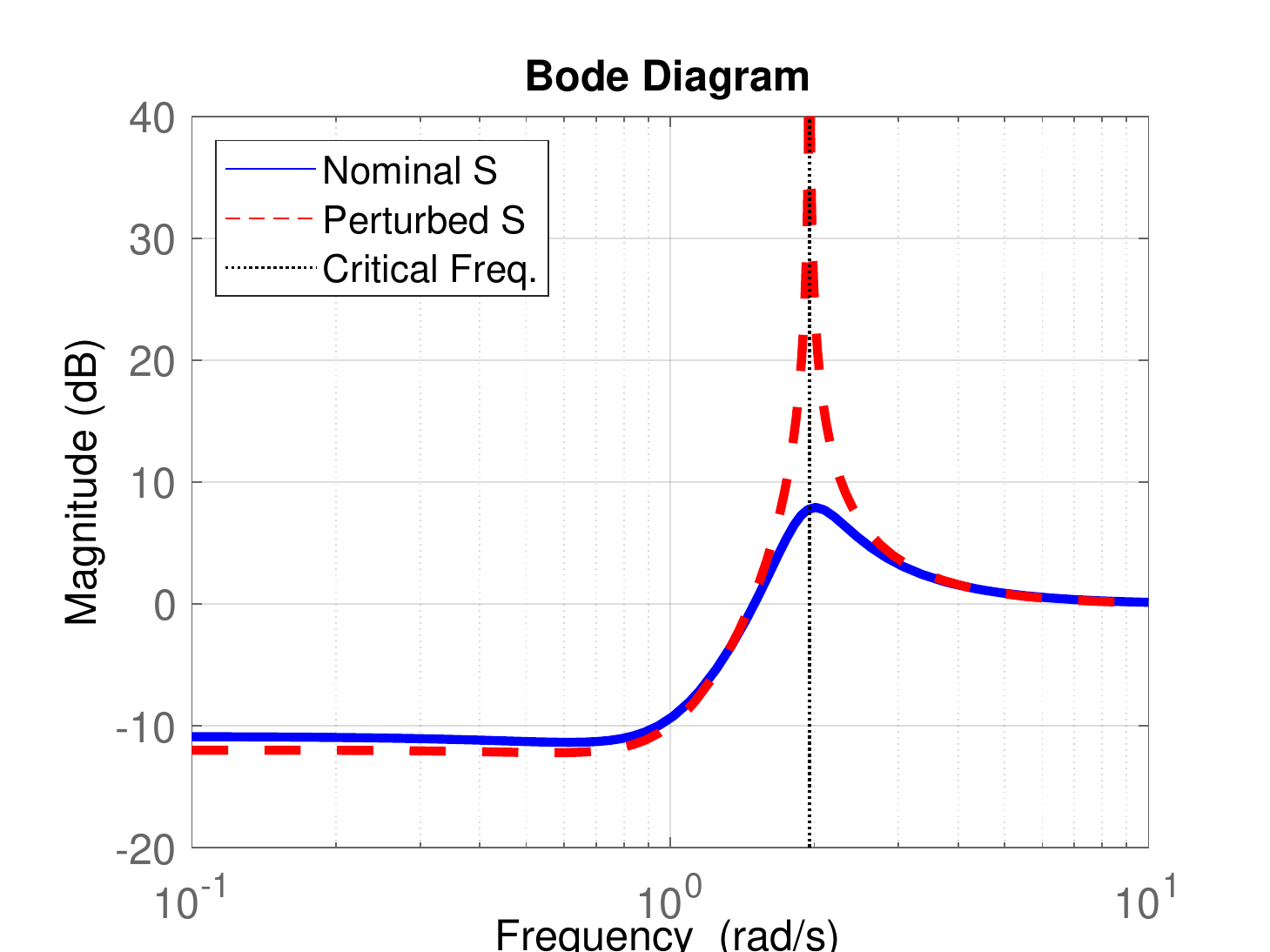}
  \caption{Bode magnitude plot of sensitivities for nominal $f=1$ and 
    destabilizing perturbation $f_0=1.128-0.483j$.}
  \label{fig:SBodemag}
\end{figure}

The destabilizing perturbation $f_0$ is a complex number with
simultaneous gain and phase variation. This critical perturbation
causes an instability with closed-loop pole on the imaginary axis at
the critical frequency $\omega_0$.  This complex perturbation $f_0$
can be equivalently represented as an LTI system with real
coefficients.  Specifically, there is a stable, LTI system $\hat f_0$
such that: (i) $\hat f_0(j\omega_0)=f_0$, and (ii) $\hat f_0(j\omega)$
remains within $D(\alpha_{max},\sigma)$ for all $\omega$. This LTI
perturbation $\hat f_0$ can be used within higher fidelity nonlinear
simulations to gain further insight.  Details on this LTI construction
are provided in the appendix entitled ``Linear Time Invariant (LTI)
Perturbations''.

\subsection{Connections to Gain and Phase Margins}

Disk margins are related to the classical notion of gain and phase
margins but provide a more comprehensive assessment of robust
stability. In particular, the uncertainty model $D(\alpha,\sigma)$ accounts
for simultaneous changes in gain and phase, whereas the classical
margins only consider variations in either gain or phase.  The disk
margin framework models gain and phase variations as a
multiplicative factor $f$ taking values in $D(\alpha,\sigma)$.
Perturbations on the unit circle ($|f|=1$) correspond to phase-only
variations while perturbations on the real axis ($f\in \R$) correspond
to gain-only variations.  The disk margin $\alpha_{max}$ can be used
to compute guaranteed gain and phase margins, denoted
$(\gamma_{min},\gamma_{max})$ and $\phi_m$ as shown in
Figure~\ref{fig:ConnectionToGMPM}.  Recall that closed-loop stability
is maintained for all $f$ in the open set $D(\alpha_{max},\sigma)$.  In
particular, the closed-loop is stable for the portions of the unit
circle and real axis that intersect the disk
$D(\alpha_{max},\sigma)$. This provides lower estimates
$(\gamma_{\min},\gamma_{\max})$ and $(-\phi_m,\phi_m)$ for the 
admissible classical gain-only and phase-only variations.  The
real-axis intercepts correspond to $\delta=\pm \alpha_{max}$ and are
given by:
\begin{align}
\label{eq:galphamax}
\gamma_{\min} = \frac{2-\alpha_{\max}(1-\sigma)}{2+\alpha_{\max}(1+\sigma)} 
\,\,\, \mbox{ and } \,\,\,
\gamma_{\max} = \frac{2+\alpha_{\max}(1-\sigma)}{2-\alpha_{\max}(1+\sigma)}.
\end{align}
To determine $\phi_m$, note that the unit circle intersects the
boundary of $D(\alpha_{max},\sigma)$ at $\cos\phi_m+j\sin\phi_m$.  Consider
the (possibly oblique) triangle formed by this intersection point, the
origin, and the center $c$ of $D(\alpha_{max},\sigma)$.  Apply the law of
cosines to this triangle to obtain
$r^2 = 1 + c^2 - 2 c\cos\phi_m$. This yields the following
expression for $\phi_m$:
\begin{align}
\label{eq:cosphim}
  \cos\phi_m = \frac{ 1+ c^2-r^2 }{2c} 
    = \frac{ 1+ \gamma_{min}\gamma_{max}}{ \gamma_{min} + \gamma_{max} }.
\end{align}
If $D(\alpha_{max},\sigma)$ fails to intersect the unit circle, e.g.
$D(\alpha_{max},\sigma)$ entirely contains the unit disk, then the right
side of Equation~\ref{eq:cosphim} will have magnitude greater than 1.
In such cases $\phi_m := +\infty$ and the feedback system is stable for
any phase variation.


\begin{figure}[h!]
  \centering
  \includegraphics[width=0.5\textwidth]{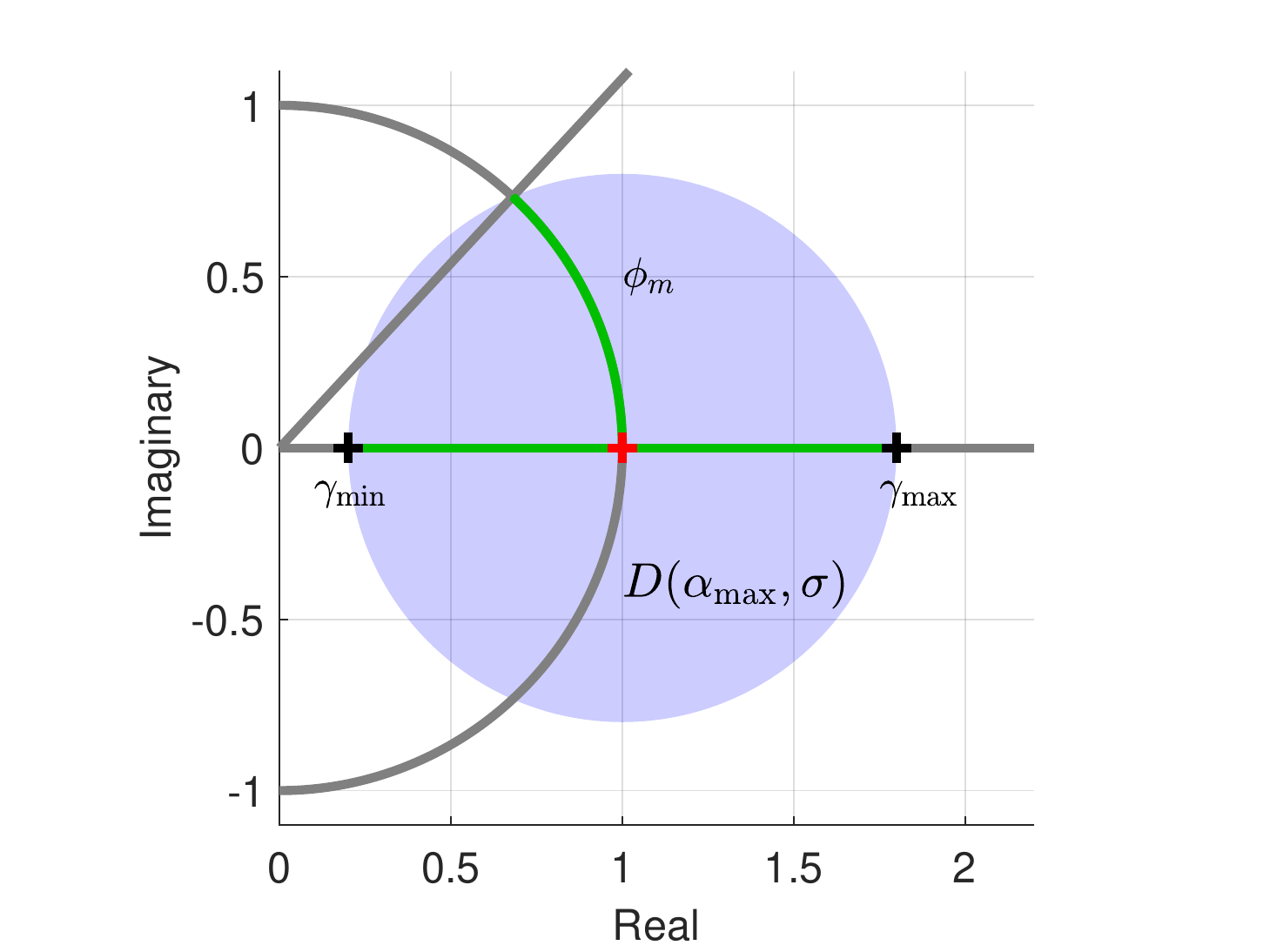}
  \caption{Guaranteed gain and phase margins from largest disk
    $D(\alpha,\sigma)$ maintaining stability.}
  \label{fig:ConnectionToGMPM}
\end{figure}

Note that $(\gamma_{\min},\gamma_{\max})$ and $(-\phi_m,\phi_m)$ are
safe levels of gain-only and phase-only variations. Each value of $\sigma$
yields a new pair of such estimates, and we can vary $\sigma$ to refine
these estimates.  This is of limited practical value, however, since
we can directly compute the classical margins and varying $\sigma$ amounts
to making assumptions on the gain variations that may not hold for the
real system. More importantly, the disk margins can be used to
quantify the effect of combined gain and phase variations that occur
in any real feedback loop. This can again be done using simple
geometry. First consider a given level $\gamma$ of gain variation as
shown in the left plot of Figure~\ref{fig:GainPhaseGeometry}. The
intercepts of the line $y = x\tan\phi$ with the bounding circle of
$D(\alpha_{max},\sigma)$ determine the safe range $(-\phi,\phi)$ for phase
variations concurrent with the gain $\gamma$.  By the law of cosines,
the value of $\phi$ satisfies
$r^2 = \gamma^2 + c^2 - 2 \gamma c \cos\phi$. This can be
equivalently expressed as:
\begin{align}
\label{eq:gammaquadratic}
\gamma^2 - \gamma (\gamma_{\min} + \gamma_{\max}) \cos \phi  +
\gamma_{\min} \gamma_{\max} = 0. 
\end{align}
This expression with gain level $\gamma=1$ simplifies to the previous
relation for $\phi_m$ (Equation~\ref{eq:cosphim}).  Next consider a
given level $\phi$ of phase variation as shown in the right plot of
Figure~\ref{fig:GainPhaseGeometry}. The intercepts of the line
$y = x \tan\phi $ with the bounding circle of $D(\alpha_{max},\sigma)$
determine the safe range $(\gamma^-,\gamma^+)$ for concurrent gain
variations. Again by the law of cosines, the values $\gamma^-$ and
$\gamma^+$ are the roots of Equation~\ref{eq:gammaquadratic} with
the phase variation $\phi$ given.

\begin{figure}[h!]
  \centering
  \includegraphics[trim={1.5cm 0 1.5cm 0},clip,width=0.42\textwidth]{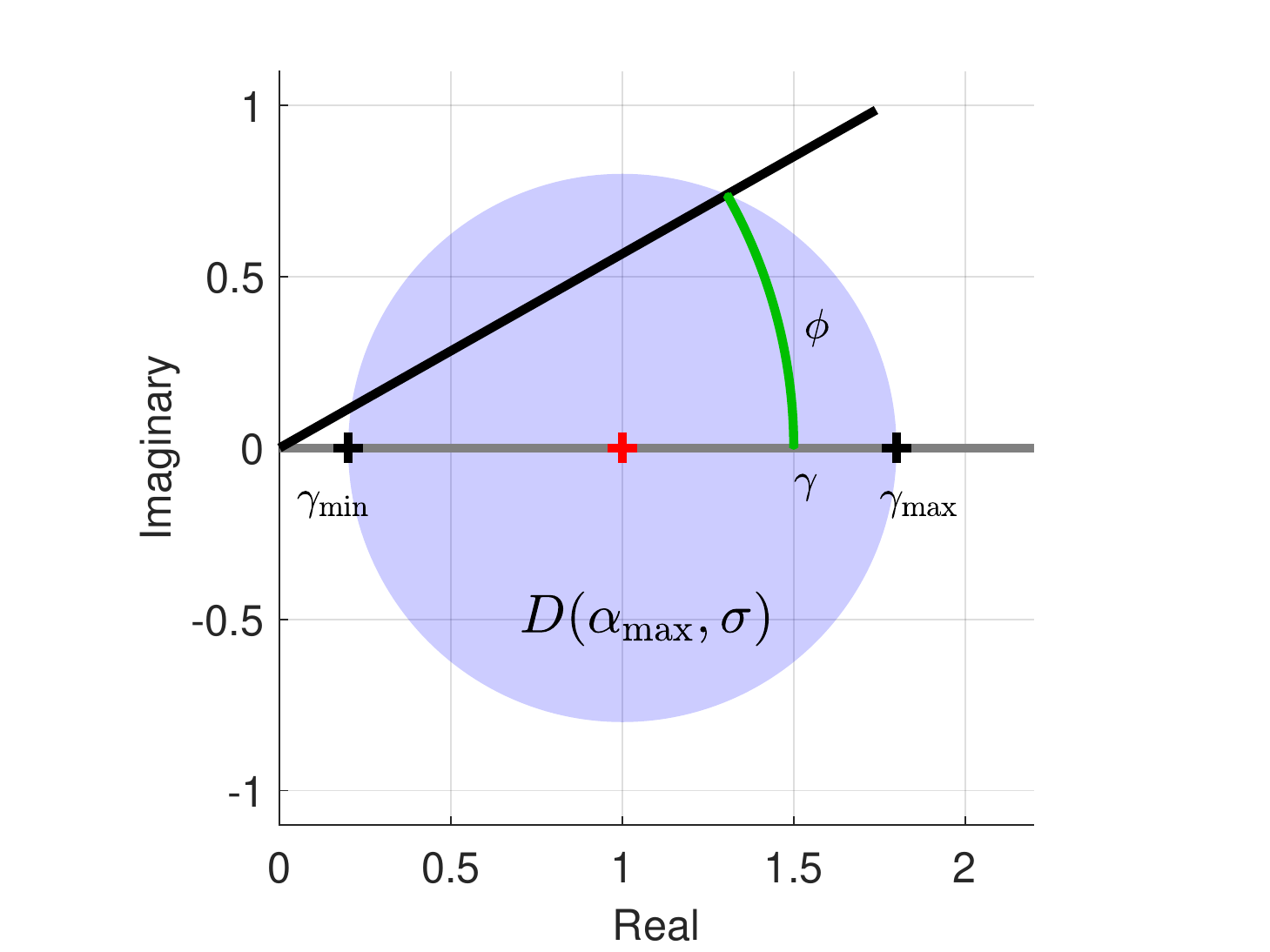}
  \includegraphics[trim={1.5cm 0 1.5cm 0},clip,width=0.42\textwidth]{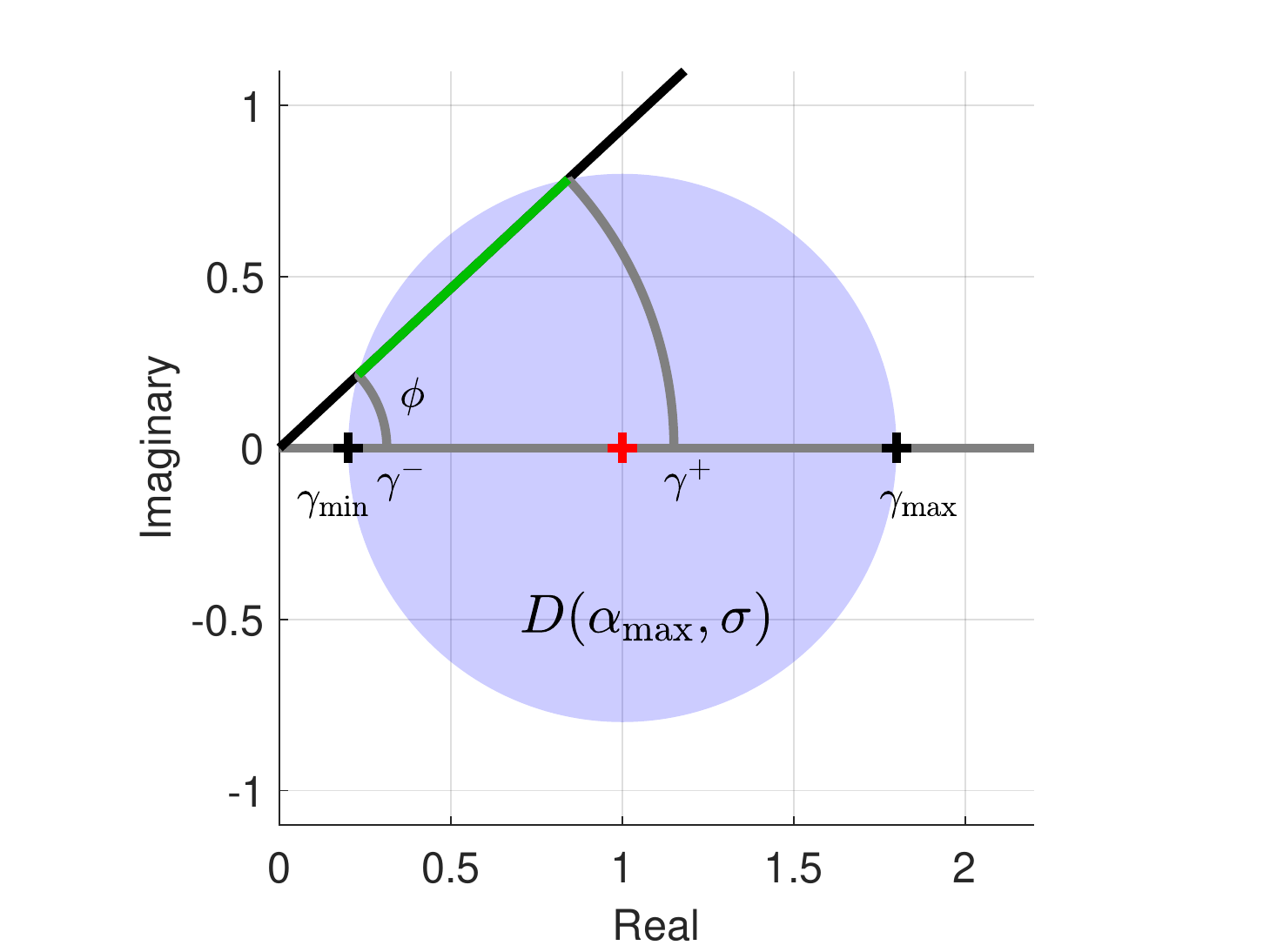}
  \caption{Geometry of admissible phase variations for a given gain
    variation $\gamma$ (Left) and admissible gain variations for a
    given phase variation $\phi$ (Right). }
  \label{fig:GainPhaseGeometry}
\end{figure}

The locus of $(\gamma,\phi)$ solutions delimits the ``safe''
variations as shown in Figure~\ref{fig:SafeGainPhase} in units of
(dB,degrees). The same bounding curve is obtained from the perturbations
$f$ corresponding to $\delta=\alpha_{max}e^{j\theta}$ with
$\theta \in [0,\pi]$. This parameterizes the bounding curve as
$(\gamma,\phi) = (|f| , {\rm angle}(f))$ with
\begin{align}
f = \frac{2+(1-\sigma) \alpha_{max} e^{j\theta}}{2-(1+\sigma) \alpha_{max} e^{j\theta}}
, \;\;\; \theta \in [0,\pi].
\end{align}
The classical gain-only and phase-only margin estimates correspond to
the boundary points $(0,\phi_m)$ and
$(20 \log_{10} \gamma_{min},20 \log_{10} \gamma_{max})$. This assumes
the standard case where the real axis intercepts satisfy
$0 < \gamma_{min} \le 1 \le \gamma_{max} <\infty$.  Recall that the
maximum phase variation $\phi_{\max}$ of any perturbation in
$D(\alpha_{max},\sigma)$ satisfies $\sin\phi_{\max} = \frac{r}{c}$ when
$r \le c$. For the balanced case $\sigma=0$ the peak phase variation occurs
at $\gamma=1$ (phase only variation) and hence $\phi_{max}=\phi_m$ for
this case. For nonzero $\sigma$, the peak $\phi_{\max}$ is not achieved for
phase-only variation and requires some amount of gain variation. The
safe region in Figure~\ref{fig:SafeGainPhase} fully quantifies how the
disk margin $\alpha_{\max}$ translates into safe levels of gain-only,
phase-only, and combined gain/phase variations.

\begin{figure}[h!]
  \centering
  \includegraphics[width=0.65\textwidth]{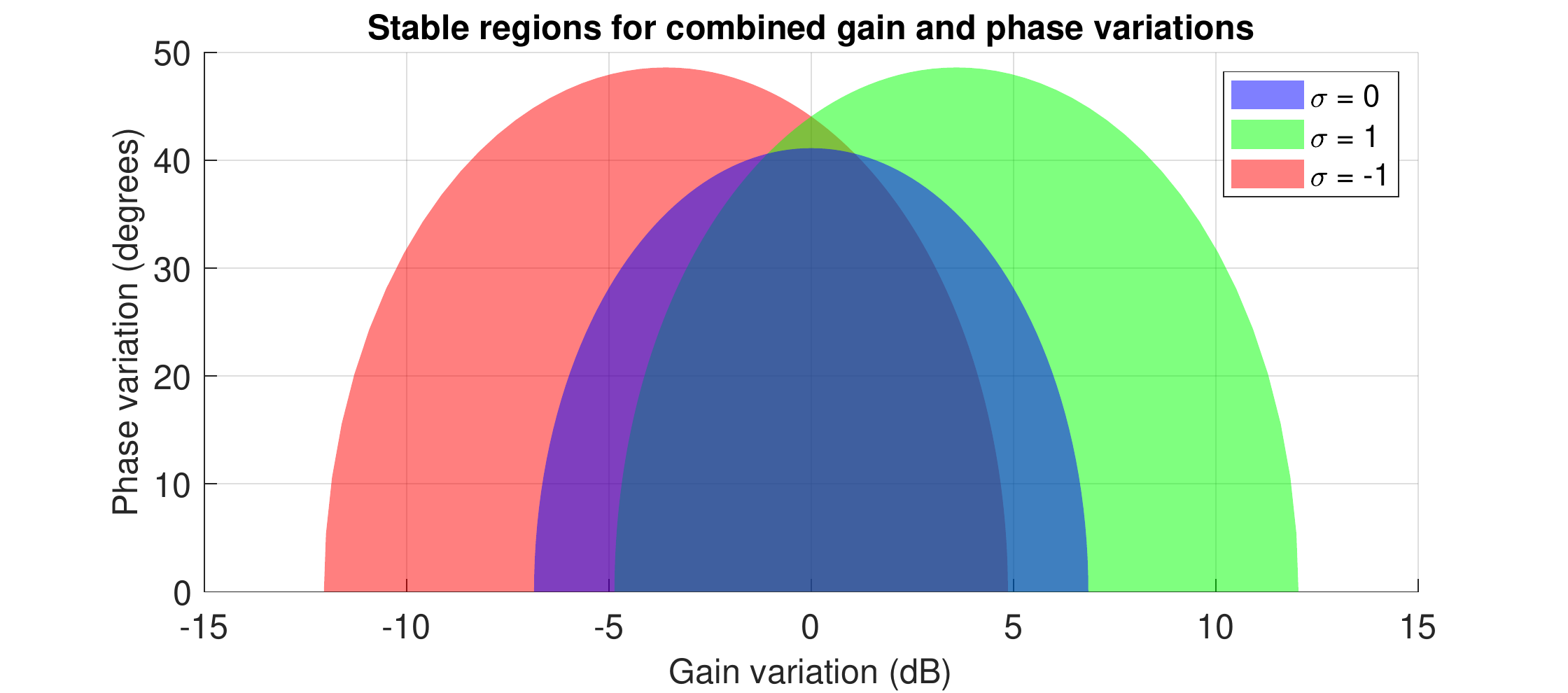}
  \caption{Safe combinations of gain and phase variations for
    $\alpha_{\max} = 0.75$.}
  \label{fig:SafeGainPhase}
\end{figure}

\begin{ex}
  \label{ex:edm2}
  The classical gain-only and phase-only margins for
  $L(s)=\frac{25}{s^3+10s^2+10s+10}$ were previously computed in
  Example~\ref{ex:cm} as $g_L=0$, $g_U =3.6$ and $\phi_U = 29.1^o$.
  Recall also that the symmetric disk margin for this loop were
  computed in Example~\ref{ex:edm} as $\alpha_{max} = 0.46$.  The
  symmetric disk provides guarantees that the classical gain margins
  are at least $g_L\le \gamma_{min} = 0.63$ and
  $g_U\ge \gamma_{max} = 1.59$.  These symmetric margins are
  $\pm 4.05$dB, i.e.  they are symmetric as multiplicative factors
  from the nominal gain of 1. The symmetric disk also guarantees
  classical phase margins of at least $\theta_U \ge \phi_m =
  25.8^o$. The gain-only and phase-only guarantees from the symmetric
  disk margin are conservative relative to the actual classical
  margins.  However, it is important to emphasize that the symmetric
  disk margin provides a stronger robustness guarantee.  Specifically,
  it ensures stability for all simultaneous gain and phase variations
  in the disk $D(\alpha_{max},\sigma=0)$.
\end{ex}

\subsection{Nyquist Exclusion Regions}


Disk margins have an interpretation in the Nyquist plane.  To
simplify the discussion, consider the typical case where
$D(\alpha_{max},\sigma)$ is the interior of a disk with real intercepts
satisfying $0<\gamma_{min}<1$ and $1<\gamma_{max}<\infty$.  The disk
margin analysis implies that $1 + f L(j \omega) \neq 0$ for all
perturbations $f\in D(\alpha_{max},\sigma)$ and all frequencies
$\omega\in \R \cup \{+\infty\}$. Rewrite this stability condition as
$L(j \omega) \neq -f^{-1}$.  The set
$\{ -f^{-1}\in \C \, : \, f \in D(\alpha_{max},\sigma) \}$ is a disk with
real axis intercepts $(-\gamma_{min}^{-1},-\gamma_{max}^{-1})$. Thus
the condition $L(j \omega) \neq -f^{-1}$ can be interpreted as a
Nyquist exclusion region, i.e. the Nyquist plot $L(j\omega)$ does not
enter the disk $\{ -f^{-1}\in \C \, : \, f \in D(\alpha_{max},\sigma) \}$.
This exclusion region contains the critical point $(-1,0)$ and is
tangent to the Nyquist curve of $L$ at some point $-1/f_0$. Varying the
skew $\sigma$ produces different exclusion regions with different
contact points.

The exclusion regions can be related to common disk margins used in
the literature.  If $\sigma=-1$ then the disk margin condition is
$\alpha_{max}=\|T\|_\infty^{-1}$. This margin is related to the robust
stability condition for models with multiplicative uncertainty of the
form $P(1+\delta)$ \cite{skogestad05,zhou96}. The real-axis intercepts
for this $T$-based margin are $\gamma_{\min} = 1-\alpha_{max}$ and
$\gamma_{\max} = 1+\alpha_{max}$.  The disk of perturbations is
centered at the nominal $f=1$ and the $\alpha_{max}$ is the radius.
The gain can increase and decrease by the same absolute amount.
However, the corresponding Nyquist exclusion disk has intercepts
$(-\gamma_{min}^{-1},-\gamma_{max}^{-1})$ and this exclusion disk is
skewed, i.e. its center is offset relative to $-1$.

If $\sigma=+1$ then the disk margin condition is
$\alpha_{max}=\|S\|_\infty^{-1}$.  The real-axis intercepts for this
$S$-based margin are $\gamma_{\min} = (1+\alpha)^{-1}$ and
$\gamma_{\max} = (1-\alpha)^{-1}$.  The disk of perturbations is
skewed with center offset from the nominal $f=1$.  The corresponding
Nyquist exclusion disk has intercepts
$(-\gamma_{min}^{-1},-\gamma_{max}^{-1})=(-1-\alpha,-1+\alpha)$.  This
Nyquist exclusion disk is centered at $-1$ with $\alpha$ as the
radius. The $S$-based margin $\alpha_{max}$ defines the distance from
the Nyquist curve of $L$ to the critical $-1$ point. Specifically, if
$\sigma=+1$ then $\alpha_{max} = \min_\omega |1+L(j\omega)|$.  Based on
this interpretation, the $S$-based margin has also been called the
vector gain margin \cite{smith58,franklin18} and modulus margin
\cite{falcoz15}.

Finally, if $\sigma=0$ then the disk margin is given by
$\alpha_{max}=\|\frac{1}{2}(S-T)\|_\infty^{-1}$. This symmetric disk
margin was introduced in \cite{barrett80} and more recently discussed
in \cite{blight94,bates02}.  The center of the perturbation disk is
offset from the nominal $f=1$ but is balanced in the sense
that $\gamma_{max} = \gamma_{min}^{-1}$. The gain variation can
increase or decrease by the same relative factor.  Moreover, the
corresponding Nyquist exclusion disk has intercepts
$(-\gamma_{min}^{-1},-\gamma_{max}^{-1})$.  This Nyquist exclusion
disk also has center offset from $-1$. However, the exclusion disk is
again balanced in the sense that the real axis intercepts are the same
relative factor from -1. Thus for $\sigma=0$ both the perturbation and
Nyquist exclusion sets are symmetric (balanced) disks.

\begin{ex}
  \label{ex:nyqexc}
  The left plot in Figure~\ref{fig:nyqexc} shows the Nyquist plot and
  three exclusion regions for
  $L(s) = \frac{25}{s^3 + 10 s^2 + 10s + 10}$.  Each exclusion region
  is the disk $\{ -f^{-1}\in \C \, : \, f \in D(\alpha_{max},\sigma) \}$
  with $\alpha_{\max}=\| S + \frac{\sigma-1}{2} \|_\infty^{-1}$. The right
  plot is zoomed more tightly on the exclusion regions. Note that each
  exclusion region is tangent to the Nyquist curve of $L$ at some
  point. These tangent points correspond to $-f_0^{-1}$ where $f_0$ is
  the destabilizing perturbation for the given skew $\sigma$.
\end{ex}

\begin{figure}[h!]
  \centering
  \includegraphics[width=0.45\textwidth]{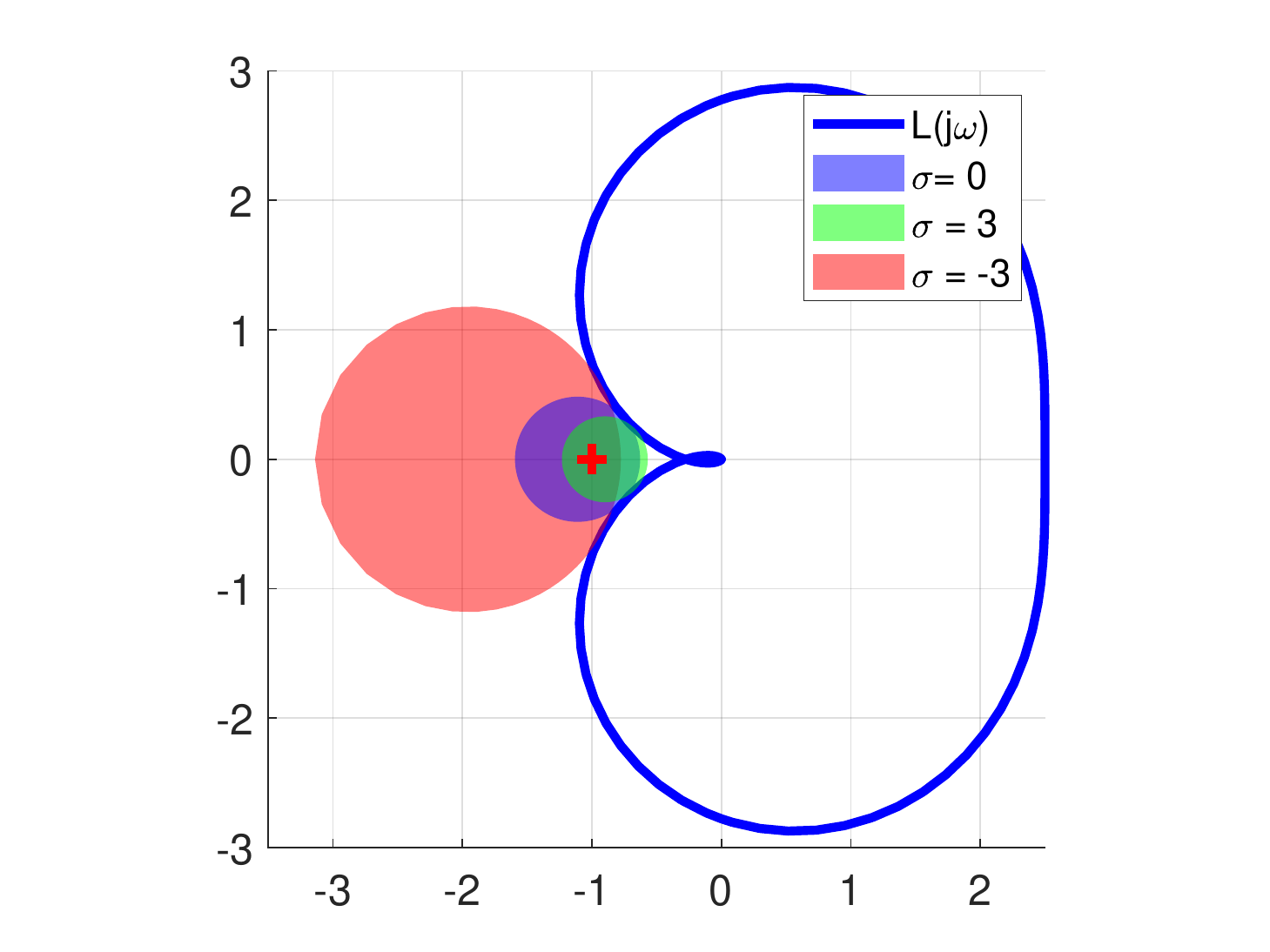}
  \includegraphics[width=0.45\textwidth]{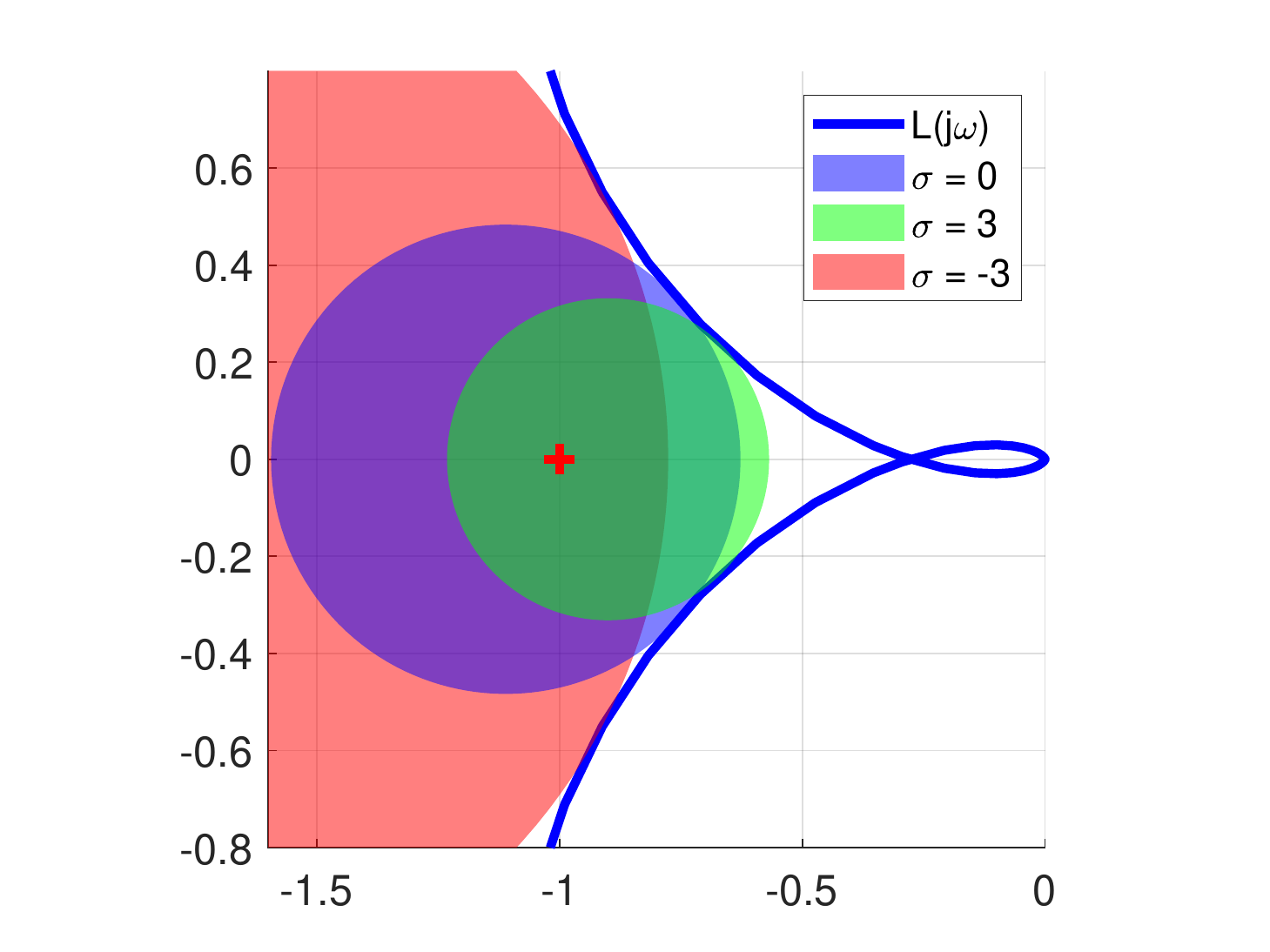}
  \caption{Nyquist exclusion regions based on disk margins with different
    skews.}
  \label{fig:nyqexc}
\end{figure}

%
%
%
%

\subsection{Frequency-Dependent Margins}

The disk margin for a given skew $\sigma$ is the largest value of
$\alpha$ such that the closed-loop remains well-posed and stable for
all perturbations in $D(\alpha,\sigma)$. The perturbations are
parameterized as $f(\delta)$ with $|\delta|<\alpha$.  Computing the
disk margin amounts to finding the smallest $\delta$ such that
$1 + f( \delta ) L(j \omega) = 0$ at some frequency $\omega$. This
problem can be considered at each frequency.  That is, define the disk
margin at the frequency $\omega$ as follows:
\begin{align}
  \alpha_{max}(\omega) 
  := \min \{ |\delta| : 1 + f( \delta ) L(j \omega) = 0 \}.
\end{align}
This specifies the minimum amount of gain and phase variation needed
to destabilize the loop at this frequency.  Similar to
Theorem~\ref{thm:edm}, this frequency-dependent margin is given by:
\begin{align}
\alpha_{max}(\omega) = \left| S(j\omega) + \frac{\sigma-1}{2}\right|^{-1}.
\end{align}
Moreover, the actual disk margin $\alpha_{max}$ is equal to the
smallest of all the frequency-dependent disk margins:
\begin{align}
  \alpha_{\max} = \min_{\omega\in \R \cup \{+\infty\}} \alpha_{max}(\omega).
\end{align}
A plot of $\alpha_{max}(\omega)$ vs. $\omega$ provides more information
about the feedback loop than just its smallest value $\alpha_{\max}$.
For example, such a plot can identify frequency bands where the disk
margin is weak. The margins in these frequency bands can then be
compared with the expected level of model
uncertainty. Frequency-dependent margins may also reveal robustness
issues away from the gain crossover frequency, e.g., near a resonant
mode that has not been sufficiently attenuated. This motivates the
case for plotting disk margins vs. frequency or, for easier
interpretation, plotting the equivalent gain-only and phase-only
margins $(\gamma_{min},\gamma_{max})$ and $\phi_m$ as a function of
frequency.  The formulas obtained earlier for $(\gamma_{min},\gamma_{max})$
and $\phi_m$ (Equations~\ref{eq:galphamax} and \ref{eq:cosphim}) can be used 
with $\alpha_{max}$ replaced by $\alpha_{max}(\omega)$.

\begin{ex}
  \label{ex:freqdepdm}
  Consider the following loop transfer function:
  \begin{align}
    L(s) = \frac{6.25 (s+3) (s+5)}{ s (s+1)^2 (s^2 + 0.18s + 100)}.
  \end{align}
  The Bode plot for this loop is shown on the left of
  Figure~\ref{fig:FreqDMForResL}. This loop has a resonance near
  10 rad/sec.  The right side of the figure plots the
  frequency-dependent gain-only and phase-only margins computed from
  the symmetric disk margin.  The gain-only plot corresponds to the
  weaker of the two gain margins,
  i.e. $\gamma_m:=\min(1/\gamma_{min},\gamma_{max})$.  At each
  frequency, the gain margin value indicates the minimum amount of
  relative gain variation needed to destabilize the loop at this
  frequency, i.e. cause a closed-loop pole to cross the imaginary axis
  at this frequency. The frequency-dependent phase margin plot has a
  similar interpretation. The frequency where these margins are
  smallest is the \emph{critical frequency} and corresponds to the
  frequency that minimizes $\alpha_{max}(\omega)$.  This pinpoints the
  frequency band where stability is most problematic and typically
  lies near the crossover frequency. The plot may also highlight other
  problematic regions. For example, the disk-based margins in
  Figure~\ref{fig:FreqDMForResL} are weak in a wide band around
  crossover but also near the first resonant mode. Also note that
  $\gamma_m \rightarrow \infty$ and $\phi_m \rightarrow 90^o$ past 10
  rad/s because $\alpha_{max}(\omega) \rightarrow 2$ and thus the
  stable region $D(\alpha_{max}(\omega),\sigma=0)$ approaches the half plane
  ${\rm Re}(f)\geq 0$.
\end{ex}

\begin{figure}[h!]
  \centering
  \includegraphics[width=0.45\textwidth]{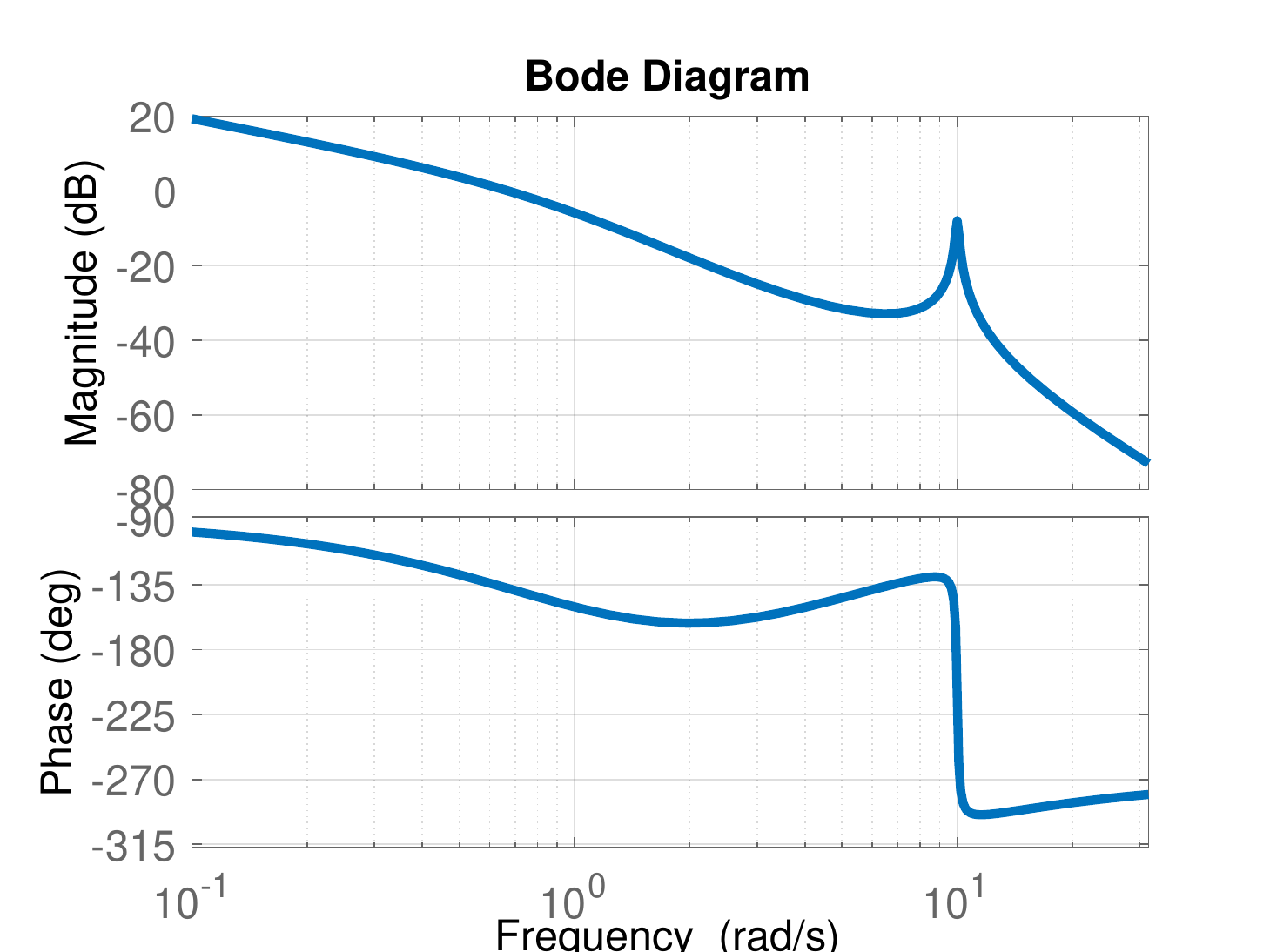}
  \includegraphics[width=0.45\textwidth]{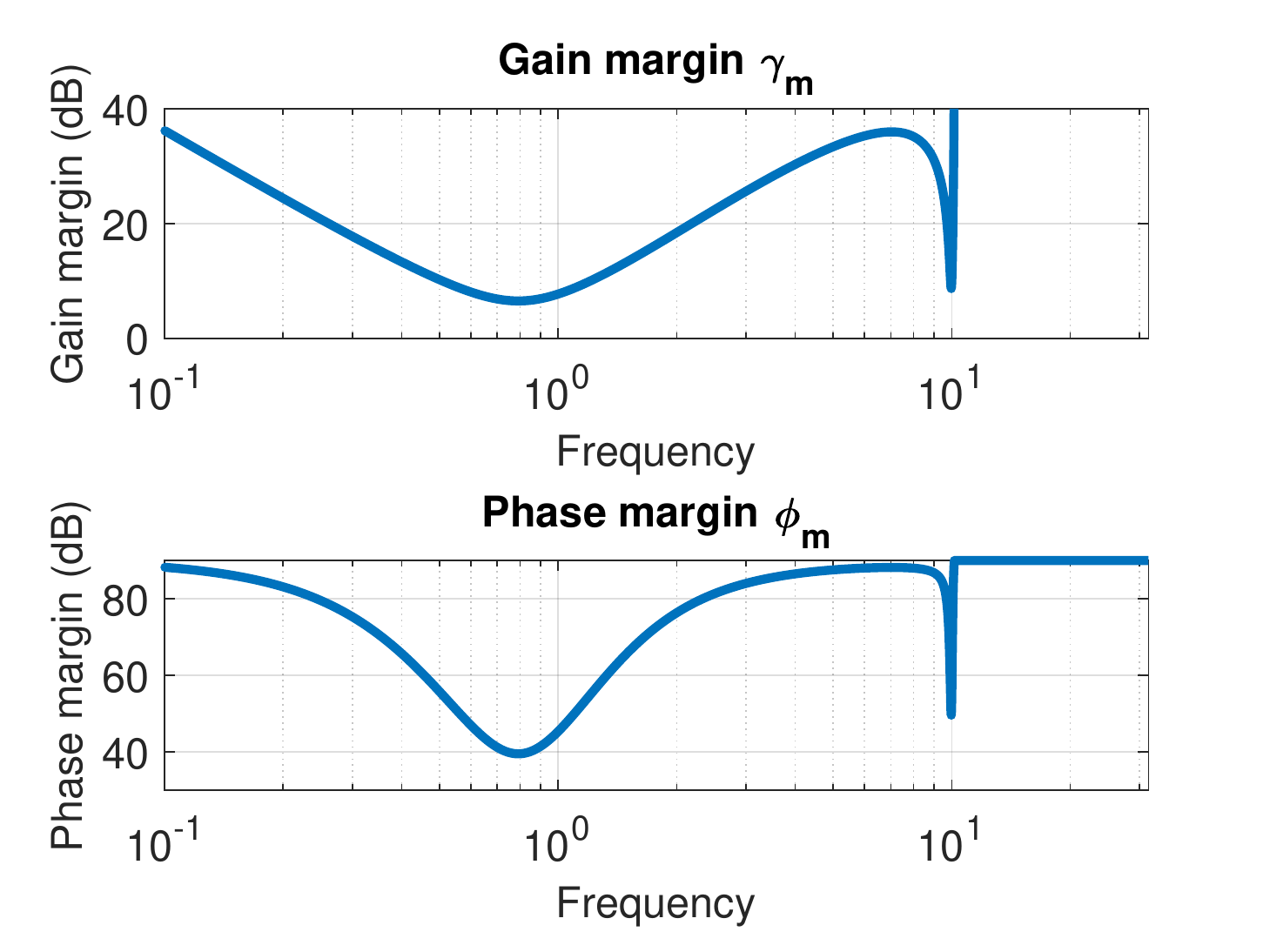}
  \caption{Open-loop response for $L$ (left) and corresponding
    frequency-dependent disk gain and phase margins for $\sigma=0$.}
  \label{fig:FreqDMForResL}
\end{figure}

\section{Margins for MIMO Systems}

This section briefly reviews two different margins for MIMO
feedback systems.  The first analysis is loop-at-a-time. This
introduces perturbations in a single channel while holding all
other channels fixed.  This can be overly optimistic as it fails
to capture the effects of simultaneous perturbations in multiple
channels.   The second analysis considers the effects of
such simultaneous perturbations in multiple channels.

\subsection{Loop-at-a-time Margins}

Loop-at-a-time analysis is a simple extension of classical margins to
assess the robustness of a MIMO feedback system.  The procedure is
illustrated for a $2\times 2$ MIMO plant as shown in
Figure~\ref{fig:MIMOpert1}.  A scalar (gain, phase, or disk)
perturbation $f_1$ is introduced at the first input of the plant
$P$. The other loop is left at its nominal (unperturbed) value.
First, break the loop at the location of the perturbation as shown on
the left side of Figure~\ref{fig:MIMOopen1}. Next, compute the
transfer function from the scalar input $z_1$ to the scalar output
$u_1$ (with the other loop closed as shown).  Denote this SISO open
loop transfer function as $L_1$. The subscript of $L_1$ reflects that
the loop was broken at the first channel at the input of $P$.  The
perturbation $f_1$ closes the loop from $u_1$ to $z_1$.  Hence the
MIMO feedback with perturbation at the first input of $P$ can be
re-drawn as the SISO feedback system shown on the right side of
Figure~\ref{fig:MIMOopen1}. The (gain, phase, or disk) margin
associated with this loop can be computed using the SISO methods
discussed previously.  This gives the margin associated with the first
input of $P$.  Note that $L_1$ is the transfer function from $z_1$ to
$u_1$ and hence Figure~\ref{fig:MIMOopen1} is in positive feedback.
The margins must be evaluated using $-L_1$ because the standard
convention assumes the loop is in negative feedback.  The margins can
be computed similarly at the second input of $P$ as well as at both
outputs of $P$.

\begin{figure}[h!]
\centering
\scalebox{0.8}{
\begin{picture}(330,90)(0,-65)
 \thicklines 
 \put(10,10){\vector(1,0){50}}  
 \put(30,-10){\vector(1,0){30}}  
 \put(60,-20){\framebox(40,40){$K$}}
 \put(100,10){\vector(1,0){40}}  
 \put(115,15){$u_1$}
 \put(140,0){\framebox(20,20){$f_1$}}
 \put(170,15){$z_1$}
 \put(160,10){\vector(1,0){40}}  
 \put(100,-10){\vector(1,0){100}} 
 \put(130,-20){$u_2=z_2$} 
 \put(200,-20){\framebox(40,40){$P$}}
 \put(240,10){\vector(1,0){80}}  
 \put(300,10){\line(0,-1){70}}  
 \put(300,-60){\line(-1,0){290}}  
 \put(10,-60){\line(0,1){70}}  
 \put(240,-10){\vector(1,0){40}}  
 \put(260,-10){\line(0,-1){30}}  
 \put(260,-40){\line(-1,0){230}}  
 \put(30,-40){\line(0,1){30}}  
\end{picture}
}
\caption{MIMO feedback system with perturbation in the first input
  channel of $P$.}
\label{fig:MIMOpert1}
\end{figure}
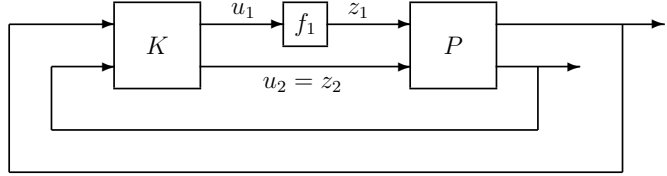

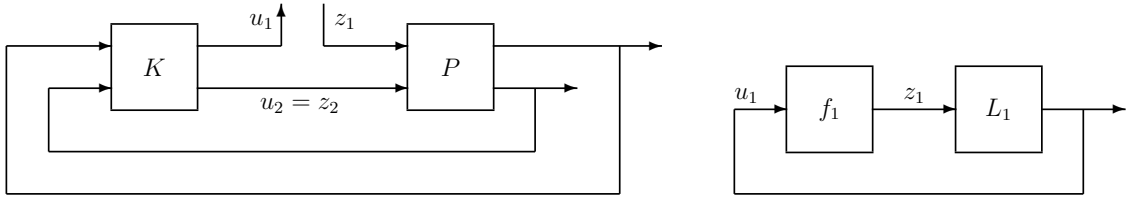
\begin{figure}[h!]
\centering
\scalebox{0.8}{
\begin{picture}(330,100)(0,-65)
 \thicklines 
 \put(10,10){\vector(1,0){50}}  
 \put(30,-10){\vector(1,0){30}}  
 \put(60,-20){\framebox(40,40){$K$}}
 \put(100,10){\line(1,0){40}}
 \put(140,10){\vector(0,1){20}}
 \put(125,20){$u_1$}
 \put(165,20){$z_1$}
 \put(160,10){\line(0,1){20}}
 \put(160,10){\vector(1,0){40}}  
 \put(100,-10){\vector(1,0){100}}  
 \put(130,-20){$u_2=z_2$} 
 \put(200,-20){\framebox(40,40){$P$}}
 \put(240,10){\vector(1,0){80}}  
 \put(300,10){\line(0,-1){70}}  
 \put(300,-60){\line(-1,0){290}}  
 \put(10,-60){\line(0,1){70}}  
 \put(240,-10){\vector(1,0){40}}  
 \put(260,-10){\line(0,-1){30}}  
 \put(260,-40){\line(-1,0){230}}  
 \put(30,-40){\line(0,1){30}}  
\end{picture}
\begin{picture}(215,70)(35,-45)
 \thicklines 
 \put(55,0){\vector(1,0){25}}  
 \put(55,5){$u_1$}
 \put(80,-20){\framebox(40,40){$f_1$}}
 \put(120,0){\vector(1,0){40}}  
 \put(135,5){$z_1$}
 \put(160,-20){\framebox(40,40){$L_1$}}
 \put(200,0){\vector(1,0){40}}   
 \put(220,0){\line(0,-1){40}}  
 \put(220,-40){\line(-1,0){165}}  
 \put(55,-40){\line(0,1){40}}  
\end{picture}
}
\caption{Left: MIMO feedback system with loop broken at the first input
  channel of $P$. \\
Right: SISO feedback with perturbation $f_1$ and loop $L_1$ obtained
at input $1$ of plant.}
\label{fig:MIMOopen1}
\end{figure}

In general, loop-at-a-time margins are computed by breaking one loop
with all other loops closed.  If the plant is $n_y \times n_u$ then
this gives $n_u$ margins at the inputs of $P$ and $n_y$ margins at the
outputs of $P$.  Unfortunately, the loop-at-a-time margins can be
overly optimistic.  In particular, a MIMO feedback system can have
large loop-at-a-time margins and yet be destabilized by small
perturbations acting simultaneously on multiple channels.  An example
is provided below to demonstrate this situation. This motivates the
development of more advanced robustness analysis tools.


%
\begin{ex}
\label{ex:sattelite1}
Consider a feedback system with the following plant and controller
with $a=10$:
\begin{align}
  P := \frac{1}{s^2+a^2} \,  
            \bmtx s-a^2 & a(s+1) \\ -a(s+1) & s-a^2 \emtx
  \mbox{ and } 
  K := -\bmtx 1 & 0 \\ 0 & 1 \emtx.
\end{align}
This example is taken from \cite{doyle78}.  The dynamics represent a
simplified model for a spinning satellite.  Additional details can be
found in Section 3.7 of \cite{skogestad05} or Section 9.6 of
\cite{zhou96}.  Breaking the loop at the first input of $P$, with the
other loop closed, yields the SISO open loop transfer function
$L_1 = -\frac{1}{s}$.  This loop (when in a positive feedback as
in Figure~\ref{fig:MIMOopen1}) has no $180^o$ phase crossover
frequencies so the classical gain margins are $g_L=0$ and
$g_U=\infty$. This loop has a single gain crossover at $\omega = 1$
rad/sec which gives a classical phase margin of
$\phi_U=90^o$. Finally, the SISO loop $L_1$ corresponds to the
sensitivity $S_1 = \frac{s}{s+1}$ and complementary sensitivity
$T_1 = \frac{1}{s+1}$.  The symmetric disk margin ($\sigma=0$) is
$\alpha_{max}=\| \frac{1}{2} (S_1-T_1) \|_\infty^{-1} = 2$. This
corresponds to a disk covering the entire RHP, i.e. stability is
maintained for any combination of gain/phase such at
$Re\{ f_1 \} > 0$.  These results demonstrate that the MIMO feedback
system is very robust to perturbations at the first input of $P$
assuming all other inputs/outputs remain at their nominal
value. Breaking the loop at the second input of $P$ or either output
of $P$ yields the same open loop transfer function, e.g.
$L_2 = -\frac{1}{s}$ at the second plant input.  Thus the
loop-at-a-time analysis demonstrates the MIMO feedback system is very
robust to perturbations at any single input or output of $P$ assuming
all other inputs/outputs remain at their nominal value.

Consider the following small simultaneous perturbation at both input
channels of the plant: $f_1 = 0.9$ and $f_2 = 1.1$.  These
simultaneous perturbations to both input channels destabilize the MIMO
feedback system. The loop-at-a-time margins fail to capture such
simultaneous variations in multiple channels.  As a consequence, the
loop-at-a-time margins provide an overly optimistic assessment of the
system robustness. 



\end{ex}








\subsection{Multi-Loop Disk Margins}



Multi-loop disk margins capture the effects of simultaneous
perturbations in multiple channels.  Figure~\ref{fig:MultiLoopMargins}
illustrates the use of multi-loop disk margins for a $2\times 2$ MIMO
plant $P$.  Scalar perturbations $f_1$ and $f_2$ are
introduced at the two input channels of the plant.  The perturbations
are restricted to a set $D(\alpha,\sigma)$ (Equation~\ref{eq:Fe}) defined
for a given skew $\sigma$. Symmetric disks of
perturbations ($\sigma=0$) are a common choice. The multi-loop disk margin
is a single number $\alpha_{max}$ defining the largest generalized disk of
perturbations $f_1$ and $f_2$ for which the closed-loop in
Figure~\ref{fig:MultiLoopMargins} is well-posed and stable.  It is
emphasized that the perturbations $f_1$ and $f_2$ are
allowed to vary independently, i.e. they are not necessarily
equal. More generally, if the plant $P$ is $n_y \times n_u$ then there
will be $n_u$ perturbations introduced at the plant input.  The margin
for this configuration is called the multi-loop \textit{input} disk
margin.  Alternatively, $n_y$ perturbations can be introduced at the
plant output.  This is referred to as the multi-loop \textit{output}
disk margin. Finally, $(n_y+n_u)$ perturbations can be introduced into
both the input and output channels to obtain the multi-loop
\textit{input/output} disk margin.

In the most general case, multi-loop margins can be defined with
perturbations introduced at arbitrary points in a feedback
system. This general formulation corresponds to a feedback system with
a collection of complex perturbations $(f_1,\ldots,f_n)$.  The
multi-loop margin is the largest value of $\alpha$ such that the
feedback system remains well-posed and stable for all perturbations
$(f_1,\ldots,f_n)$ in the set $D(\alpha,\sigma)$ specified for a given
skew $\sigma$.  The next two examples illustrate various types of
multi-loop margins.  The theory required to compute such multi-loop
margins is reviewed below in the subsection entitled ``Computing
Multi-Loop Disk Margins''.

\begin{figure}[h!]
\centering
\scalebox{0.8}{
\begin{picture}(330,90)(0,-65)
 \thicklines 
 \put(10,10){\vector(1,0){50}}  
 \put(30,-10){\vector(1,0){30}}  
 \put(60,-20){\framebox(40,40){$K$}}
 \put(100,10){\vector(1,0){40}}  
 \put(115,17){$u_1$}
 \put(140,2){\framebox(20,20){$f_1$}}
 \put(170,17){$z_1$}
 \put(160,12){\vector(1,0){40}}  
 \put(100,-12){\vector(1,0){40}}  
 \put(115,-7){$u_2$}
 \put(140,-22){\framebox(20,20){$f_2$}}
 \put(170,-7){$z_2$}
 \put(160,-12){\vector(1,0){40}}  
\put(200,-20){\framebox(40,40){$P$}}
 \put(240,10){\vector(1,0){80}}  
 \put(300,10){\line(0,-1){70}}  
 \put(300,-60){\line(-1,0){290}}  
 \put(10,-60){\line(0,1){70}}  
 \put(240,-10){\vector(1,0){40}}  
 \put(260,-10){\line(0,-1){30}}  
 \put(260,-40){\line(-1,0){230}}  
 \put(30,-40){\line(0,1){30}}  
\end{picture}
}
\caption{Multi-loop input disk margins for a $2\times 2$ plant $P$.}
\label{fig:MultiLoopMargins}
\end{figure}
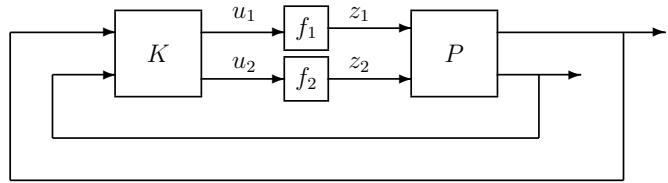


\begin{ex}
\label{ex:sattelite2}
Consider the spinning satellite discussed in
Example~\ref{ex:sattelite1}. The multi-loop input margin is computed for
this $2 \times 2$ feedback system using symmetric disks ($\sigma=0$).  This
yields $\alpha_{max}=0.0997$ corresponding to the disk with
$\gamma_{max} = \frac{1+0.5\alpha_{max}}{1-0.5\alpha_{max}} = 1.105$ and
$\gamma_{min}=\gamma_{max}^{-1} = 0.905$.  Hence the plant can
tolerate independent perturbations $f_1$ and $f_2$ at the plant inputs
with gain-only variations in $(0.905,1.105)$.  These margins
indicate that the spinning satellite feedback system is sensitive to
small perturbations occurring at both inputs to the plant. The
multi-loop output margin is the same for this system.  Multi-loop
margins can also be defined with perturbations introduced
(simultaneously) at the two inputs and two output channels.  For the
spinning satellite, this multi-loop input/output margin is
$\alpha_{max} = 0.0498$ corresponding to
$(\gamma_{min},\gamma_{max})= (0.941,1.051)$. Details on this example
including corresponding code can be found in the Matlab example
entitled ``MIMO stability margins for spinning satellite''.
\end{ex}

\begin{ex}
\label{ex:sim}
Consider the Simulink diagram for an aircraft longitudinal controller
shown in Figure~\ref{fig:airframe}.  The left side of the figure shows
blocks for the airframe dynamics, inner loop pitch-rate ($q$) control,
and outer-loop vertical acceleration ($a_z$) control.  The right side
of the figure shows one the subsystem containing the aerodynamics for
the airframe model. This Simulink model is part of a Matlab example
entitled ``Stability Margins of a Simulink Model''.  The model is
modified to include three complex perturbations inserted at various
points.  One perturbation is inserted at the plant input (red dot on
left diagram). Two other perturbations are inserted in the
aerodynamics subsystem (red dots on right diagram). These are inserted
on signals for the vertical force $F_z$ and pitching moment $M$.
These two additional perturbations can be used to model, for example,
the discrepancy in the modeled and actual aerodynamics for this force
and moment.

Figure~\ref{fig:simulinkmargin} shows the Matlab code to compute two
different disk margins for this example. The \texttt{linio} command
specifies the analysis points.  The model is nonlinear and hence the
dynamics must first be linearized around an operating point. This is
done with the \texttt{linearize} command. The symmetric disk margin is
computed at the plant input (\texttt{DMi}). Note that
\texttt{linearize} returns the loop transfer function assuming
positive feedback while \texttt{diskmargin} assumes negative feedback.
This symmetric disk margin at the plant input is
$\alpha_{max} = 0.774$.  This corresponds to a disk with
$(\gamma_{min},\gamma_{max})=(0.442,2.263)$. Hence the classical
margins are at least $g_L \le 0.442$, $g_U\ge 2.263$ and
$\phi_U \ge 42.3^o$. Next the disk margins are computed using all
three analysis points.  The multi-loop margin with symmetric disks
(\texttt{MM3}) is $\alpha_{max} = 0.428$.  Hence the feedback system
remains well-posed and stable for independent perturbations at the
three analysis points that remain in the disk with
$(\gamma_{min},\gamma_{max})= (0.648,1.544)$.

\begin{figure}[h!]
  \centering
  \includegraphics[width=0.6\textwidth]{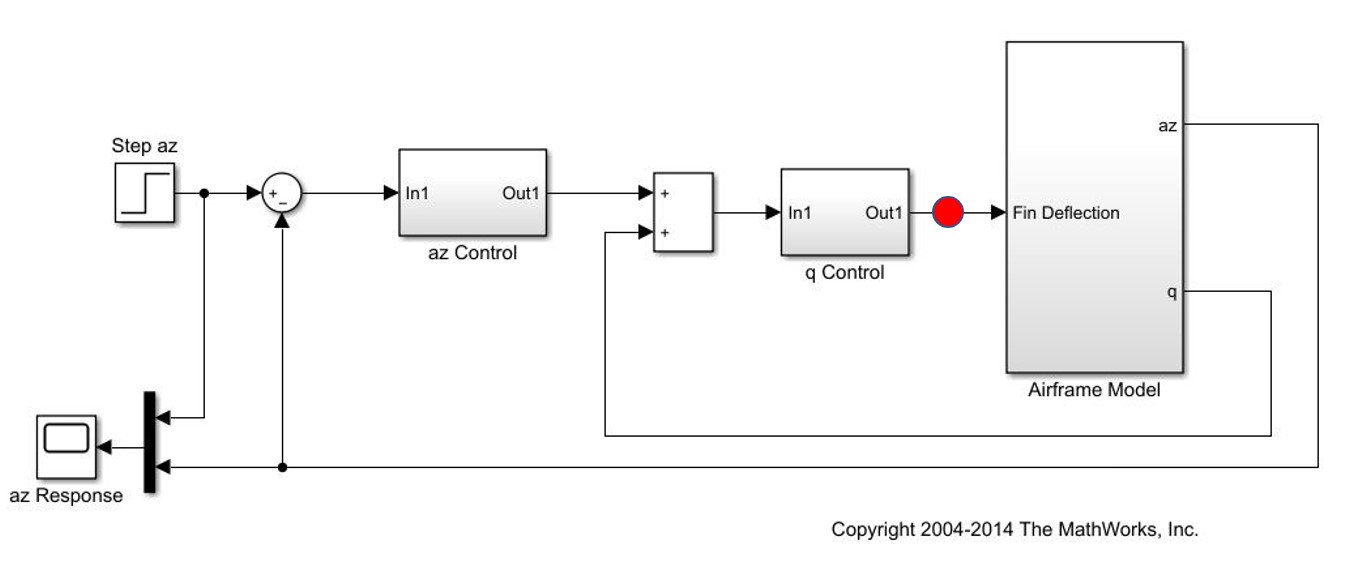}
  \includegraphics[width=0.38\textwidth]{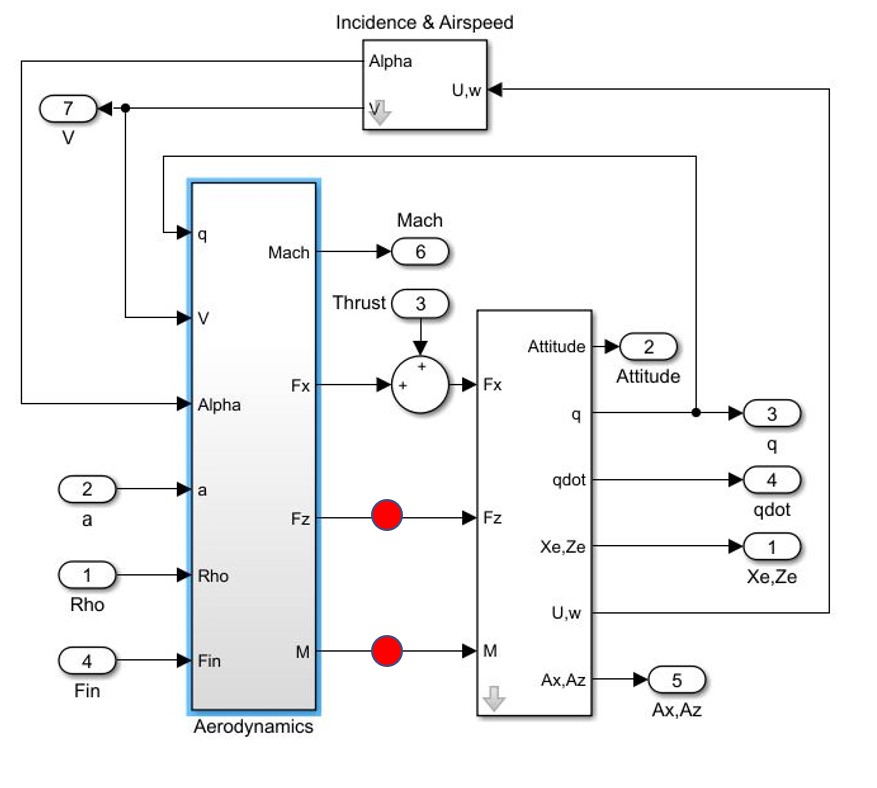}
  \caption{Simulink diagram for a longitudinal aircraft controller.}
  \label{fig:airframe}
\end{figure}
\end{ex}

\begin{figure}[t!]
\begin{footnotesize}
\begin{verbatim}
% Open simulink model from Matlab example
open_system('airframemarginEx.slx')

% Specify analysis point at plant input
aPoints(1) = linio('airframemarginEx/q Control',1,'looptransfer');

% Specify analysis points inside aerodynamic model
blk = ['airframemarginEx/Airframe Model/' ...
   'Aerodynamics & Equations of Motion/Aerodynamics'];
aPoints(2) = linio(blk,3,'looptransfer');
aPoints(3) = linio(blk,4,'looptransfer');

% Linearize and compute disk margin at plant input
Li = linearize('airframemarginEx',aPoints(1) );   
DMi = diskmargin(-Li)

% Linearize and compute disk margins at three analysis points
L3 = linearize('airframemarginEx',aPoints);
[DM3,MM3] = diskmargin(-L3)
\end{verbatim}
%
%
%
\end{footnotesize}
\caption{Code for aircraft multi-loop margins.}
\label{fig:simulinkmargin}
\end{figure}

\subsection{Computing Multi-Loop Disk Margins}

Consider a feedback system with $n$ complex perturbations
$(f_1,\ldots,f_n)$ introduced at arbitrary points. It is assumed that
the feedback system is well-posed and stable if all perturbations are
at their nominal value, $f_i=1$ for all $i$.  The multi-loop disk
margin, denoted $\alpha_{max}$, was defined in the subsection entitled
``Multi-Loop Disk Margins''. It is the largest value of $\alpha$ such
that the feedback system remains well-posed and stable for all
perturbations $(f_1,\ldots,f_n)$ in the set $D(\alpha,\sigma)$ with a given
disk skew $\sigma$.

The condition for SISO disk margins (Theorem~\ref{thm:edm}) can be
generalized for the multi-loop case.  The starting point for the SISO
disk margin result was the condition: $f\in D(\alpha,\sigma)$ places a
closed-loop pole at $s=j\omega$ if and only if $1+fL(j\omega)=0$.  The
next step was to express the perturbation $f$ in terms of
$|\delta|<\alpha$.  This led to the following stability condition
(Equation~\ref{eq:SdmCondition}):
\begin{align}
  \label{eq:SdmCondition2}
  1-\delta \left( S(j\omega) + \frac{\sigma-1}{2} \right) =0.
\end{align}
This has the form $1-\delta M(j\omega)=0$ where $M:=S+\frac{\sigma-1}{2}$.
This is the stability condition for a feedback system with $\delta$ in
positive feedback with $M$.  Similarly, each perturbation in a
multi-loop analysis can be expressed as
$f_i = \frac{2+(1-\sigma)\delta_i}{2-(1+\sigma) \delta_i}$ for some
$|\delta_i|<\alpha$. In this way the multi-loop margin analysis
involving perturbations $f_i$ is mapped to an equivalent $M$-$\Delta$
positive feedback loop as shown in Figure~\ref{fig:MDelta}.  Here $M$
is a stable $n \times n$ system and $\Delta \in \C^{n \times n}$ is
the diagonal matrix of complex perturbations
$\Delta:=diag(\delta_1,\ldots,\delta_n)$.  The multi-loop margin is
equivalent to the largest value of $\alpha$ such that the positive
feedback system with $M$ and $\Delta:=diag(\delta_1,\ldots,\delta_n)$
is well-posed and stable for all complex perturbations
$|\delta_i| < \alpha$ ($i=1,\ldots,n$). Additional details on this
$M$-$\Delta$ modeling framework can be found in
\cite{zhou96,dullerud00,skogestad05}.

\begin{figure}[h]
\centering
\tikzstyle{block} = [draw, rectangle, minimum width=1.2cm, 
     minimum height=1.2cm, align=center]
\tikzstyle{sum} = [draw,circle,inner sep=0mm,minimum size=3mm] 
\begin{tikzpicture}[auto,>=stealth',very thick,node distance = 1.5cm]
\node [coordinate, name=input1] {};
\node[sum, right = of input1](sum1){};
\node[block, right=of sum1](delta) {$\Delta$};
\node[block, below=0.8cm of delta](M) {$M$};
\node[sum, right = of M](sum2){};
\node [coordinate, right=of sum2](input2) {};
\draw[->] (input1) -- (sum1.west) node[pos=0.2]{$d_1$};
\draw[->] (sum1.east) -- (delta.west)  node[pos=0.5]{$e_1$};
\draw[->] (delta.east) -| (sum2.north);
\draw[->] (input2) -- (sum2.east) node[pos=0.2]{$d_2$};
\draw[->] (sum2.west) -- (M.east)  node[pos=0.5]{$e_2$};
\draw[->] (M.west) -| (sum1.south) node[pos=0.3]{$v$};
\end{tikzpicture} 
\caption{$M$-$\Delta$ feedback system for multi-loop margins.}
\label{fig:MDelta}
\end{figure}
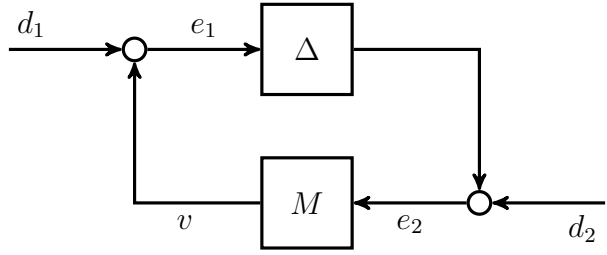


The nominal perturbation corresponds to $\Delta=0$ with nominal system
$M$.  The assumption of nominal stability thus implies the poles of
$M$ are in the LHP. The perturbation $\Delta$ causes the closed-loop
poles to move continuously in the complex plane away from their
nominal values.  The poles may move into the RHP (unstable
closed-loop) if $\Delta$ is varied by a sufficiently large amount from
the nominal value $\Delta=0$. The transition from stable to unstable
occurs when the closed-loop poles cross the imaginary axis. As in the
SISO case, it is thus useful to have a condition that characterizes
this stability transition, i.e. a condition that characterizes the
existence of imaginary axis poles. It can be shown that the
$M$-$\Delta$ system has a pole on the imaginary axis at $j\omega$ if
and only if $ \det (I - M(j\omega) \Delta )=0$.  To sketch a
simplified derivation, consider the case where $M$ has no direct
feedthrough $(D=0)$. Let $(A,B,C,D=0)$ be a state-space realization
for $M$. The poles of the $M$-$\Delta$ system are given by the
eigenvalues of the state matrix $A_{cl}:=A+B \Delta C$.  There is a
pole on the imaginary axis at $j\omega$ if and only if
$\det( j\omega I - A_{cl})=0$.  Stability of $M$ implies that
$j\omega$ is not an eigenvalue of $A$.  Hence $(j\omega I-A)$ has a
non-zero determinant and its inverse exists. Thus
$\det( j\omega I - A_{cl})=0$ is equivalent to
$\det\left( I - (j\omega I-A)^{-1} B \Delta C \right)=0$.  Finally,
apply Sylvester's determinant identity (Corollary 3.9.5 in
\cite{bernstein18}):
\begin{align*}
  0 = \det (I - (j\omega I-A)^{-1} B \Delta C ) =   
  \det (I - M(j\omega) \Delta).
\end{align*}
If there is only one perturbation ($n=1$) then the determinant
condition simplifies to $1-M(j\omega) \cdot \delta=0$.  This is the
same condition that appeared in the proof for the SISO Small Gain
result (Equation~\ref{eq:SdmCondition} and rewritten in
Equation~\ref{eq:SdmCondition2}).

In this SISO case, if the gain $|M(j\omega)|$ is large then there is a
small perturbation $\delta=M(j\omega)^{-1}$ that causes a pole on the
imaginary axis at $j\omega$. The MIMO case requires an appropriate
generalization for the connection between the ``gain'' of the system
$M$ and the existence of small, destabilizing perturbations
$(\delta_1,\ldots,\delta_n)$.  First, let
$\mathbf{\Delta} \subset \C^{n\times n}$ denote the set of diagonal,
complex matrices and define the norm for any
$\Delta \in \mathbf{\Delta}$ by
$\|\Delta\|:=\max_{i=1,\ldots,n} | \delta_i |$.  In other words, the
norm is given by the largest (magnitude) of the diagonal entries.
Note that all perturbations $f_i$ are in the set $D(\alpha,\sigma)$
if and only if $|\delta_i|<\alpha$, i.e. if and only if
$ \| \Delta \| < \alpha$.

Next, define the function
$\mu: \C^{n \times n} \rightarrow [0,\infty)$ by:
\begin{align}
  \label{eq:ssv}
  \mu(M_0) := \left( \min_{\Delta \in \mathbf{\Delta}} \|\Delta\| \, : \,
  \det( I - M_0 \Delta ) = 0 \right)^{-1}.
\end{align}
By definition, $\mu(M(j\omega))$ is large if and only if there is a
``small'' $\Delta_0 \in \mathbf{\Delta}$ such that
$\det(I-M(j\omega)\Delta_0)=0$. By the discussion above, this
perturbation causes the $M$-$\Delta$ system to have a pole on the
imaginary axis.  This function $\mu$ is known as the structured
singular value or simply ``mu''
\cite{doyle78,safonov80,doyle82,doyle85,packard93,fan91}.  The
structured singular value can be used to assess robust stability and
performance of systems with more general types of uncertainties
including real, complex, and dynamic LTI uncertainties. The version in
Equation~\ref{eq:ssv} is a special instance of this more general
framework adapted for multi-loop disk margins.  It is difficult to
exactly compute $\mu(M_0)$ for a given complex matrix $M_0$ and
uncertainty set $\mathbf{\Delta}$.  However, there are efficient
algorithms to compute upper and lower bounds on $\mu(M_0)$.  The
following theorem provides a condition for the multi-loop disk margin
using this function $\mu$. It uses the following notation for the peak
value of $\mu$ across all frequencies:
\begin{align}
  \| \mu(M) \|_\infty:= \max_{\omega \in \R \cup \{+\infty\}}
  \mu\left( M(j\omega) \right).
\end{align}

\begin{theorem}
\label{thm:MLdm}
Assume $M$ is proper and stable. The multi-loop disk margin is given
by $\alpha_{max} = \| \mu(M) \|_\infty^{-1}$.
\end{theorem}
\begin{proof}
  The proof consists of two steps.  First, it is shown that there is a
  destabilizing perturbation on the boundary of the disk
  $|\Delta| < \|\mu(M)\|_\infty^{-1}$.  Let $\omega_0$ be the
  frequency (possibly infinite) where $\mu(M(j\omega))$ achieves its
  peak. By definition, there is a perturbation $\Delta_0$ such that
  (i) $\det(I-M(j\omega_0)\Delta_0)=0$, and (ii)
  $\|\Delta_0\| =\|\mu(M)\|_\infty^{-1}$. The $M$-$\Delta$ system is
  either ill-posed ($\omega_0$ infinite) or unstable with an imaginary
  axis pole ($\omega_0$ finite).  Any open disk with radius larger
  than $\|\mu(M)\|_\infty^{-1}$ contains this destabilizing
  perturbation.  Hence the multi-loop disk margin is
  $\le \|\mu(M)\|_\infty^{-1}$.  

  Next, it is shown that the $M$-$\Delta$ feedback system is stable
  and well-posed for all perturbations in the interior of
  $\|\Delta\| < \|\mu(M)\|_\infty^{-1}$. It follows from the definition
  of $\mu$ that the $M$-$\Delta$ system is well-posed and has no
  imaginary axis poles for any perturbation
  $\|\Delta\| < \|\mu(M)\|_\infty^{-1}$. Hence the closed-loop is stable
  for all $\|\Delta\|< \|\mu(M)\|_\infty^{-1}$ because the poles do not cross
  the imaginary axis into the RHP. This can be formalized with a
  homotopy argument.
\end{proof}

Additional details on computing disk margins using the structured
singular value can be found in \cite{deodhare98,bates01}.  The
structured singular value can be used extend the results in this paper
for assessing robust stability and performance with more general
classes of parametric and dynamic uncertainty. The integral quadratic
constraint framework \cite{megretski97} is even more general and can
be used to assess the impact of nonlinearities.

\section{Conclusion}

This paper provided a tutorial introduction to disk margins. These are
robust stability measures that account for simultaneous gain and phase
perturbations in a feedback system. They can also be used to compute
frequency-dependent margins which provide additional insight into
potential robustness issues. Disk margins were also described for
multiple-loop analysis of MIMO systems.  This multiple-loop analysis
provides a more accurate robustness assessment than loop-at-a-time
analysis. These multiple-loop disk margins also provide an
introduction to more general robustness frameworks, e.g.  structured
singular value $\mu$ and integral quadratic constraints.

\section*{Acknowledgment}

The authors thank Christopher Mayhew, Raghu Venkataraman, and Brian
Douglas for helpful suggestions.  The authors also gratefully
acknowledge Brian Douglas for the creation of a tutorial video
corresponding to this paper. Finally, the authors  thank Seagate for
providing the hard disk drive frequency responses shown in
Figure~\ref{fig:GMPMLimitations}.


\bibliographystyle{IEEEtran}
\bibliography{DiskMarginRefs}


\sidebars 

\clearpage
\newpage

\section{Proof of Disk Margin Condition}

This appendix proves the main technical result used to compute disk
margins (Theorem~\ref{thm:edm}).  It is assumed for simplicity, that
$L$ has no feedthrough, i.e. $D=0$. The results require some minor
modifications for systems with non-zero feedthrough, e.g. to handle
well-posedness. First, the stability transition condition is stated as
a technical lemma with a formal proof using state-space
arguments.

\begin{lemma}
\label{lem:stabtrans}
Assume the closed-loop is stable for a nominal, SISO loop $L$.  In
addition, let $\omega_0$ be a given frequency and assume
$L(j\omega_0)\ne 0$.  There is a perturbation $f_0 \in D(\alpha,\sigma)$
that causes the closed-loop to have a pole at $s=j\omega_0$ if and
only if $\left( S(j\omega_0) + \frac{\sigma-1}{2} \right) \delta_0 = 1$
holds for some $|\delta_0| < \alpha$.
\end{lemma}
\begin{proof}
  Let $(A,B,C,D=0)$ denote a state-space representation for then
  nominal loop $L$. Let $T_f$ denote the transfer function from
  reference $r$ to output $y$ for the perturbed feedback system in
  Figure~\ref{fig:gpmfb}, i.e. the complementary sensitivity
  function. The notation $T$ with no subscript will refer to the
  nominal complementary sensitivity with $f=1$.
  
  A state-space realization for the perturbed $T_f$ is given by
  $(A-fBC,B,C,0)$.  Hence the condition for some $f_0 \in D(\alpha,\sigma)$
  to cause a closed-loop pole at $s=j\omega_0$ is:
  \begin{align}
    \label{eq:detcond1}
    \begin{split}
      0 & =\det\left( j\omega_0 I - (A-f_0BC) \right) \\
      & = \det\left( j\omega_0 I-(A-BC) + (f_0-1)BC \right).
      \end{split}
  \end{align}
  The second equality simply groups the state matrix $(A-BC)$ for the
  nominal closed-loop with $f=1$.  The nominal closed-loop is assumed
  to be stable and thus $j\omega_0 I - (A-BC)$ is nonsingular.  Hence
  the Equation~\ref{eq:detcond1} is equivalent to:
  \begin{align}
    \label{eq:detcond2}
    0 = \det\left( I + 
         (f_0-1) \, \left(j\omega_0 I-(A-BC)\right)^{-1} BC \right).
  \end{align}
  Finally, apply Sylvester's determinant identity (Corollary 3.9.5 in
  \cite{bernstein18}) to shift around $C$ and obtain:
  \begin{align}
    \label{eq:detcond3}
    0 = 1 + (f_0-1) \, C (j\omega_0 I - (A-BC))^{-1} B  
      = 1 + (f_0-1) \, T(j\omega_0).
  \end{align}
  (As an aside, note that $T=\frac{L}{1+L}$ and hence
  Equation~\ref{eq:detcond3} is equivalent to
  $1+f_0 L(j\omega_0) = 0$.)  The perturbation can be expressed as
  $f_0=\frac{2+ (1-\sigma)\delta_0 }{2- (1+\sigma)\delta_0}$ for some
  $\delta_0\in \C$ with $|\delta_0|<\alpha$
  (Equation~\ref{eq:Fe}). Thus Equation~\ref{eq:detcond3} can be
  re-written, after some algebra, in terms of the nominal sensitivity
  $S:=\frac{1}{1+L}$ as follows:
  \begin{align}
    \left( S(j\omega_0) + \frac{\sigma-1}{2} \right) \delta_0 = 1.
  \end{align}
  This final step requires the assumption that $L(j\omega_0)\ne 0$.
  This ensures $S(j\omega_0) \ne 1$ and $\delta_0 \ne \frac{1+\sigma}{2}$
  so that the corresponding perturbation $f_0$ is finite.
\end{proof}


The main disk margin condition (Theorem~\ref{thm:edm}) is restated
below with a formal proof. This is a variation of a technical result
known as the small gain theorem \cite{zhou96,dullerud00,skogestad05}.

\newtheorem*{thm:edm}{Theorem \ref{thm:edm}}
\begin{thm:edm}[Restated]
  Let $\sigma$ be a given skew parameter defining the disk margin.
  Assume the closed-loop is well-posed and stable with the nominal,
  SISO loop $L$. Then the disk margin is given by:
\begin{align*}
  \alpha_{max} = \frac{1}{\left\| S + \frac{\sigma-1}{2} \right\|_\infty}
  \tag{\ref{eq:alphadm}, Restated}
\end{align*}
\end{thm:edm}
\begin{proof}
  Define $\alpha_0:=\|S + \frac{\sigma-1}{2}\|_\infty^{-1}$.  The proof
  consists of two steps. First, it is shown that there is a
  destabilizing perturbation on the boundary of $D(\alpha_0,\sigma)$.  The
  perturbation set $D(\alpha,\sigma)$ contains this destabilizing
  perturbation for any value $\alpha \ge \alpha_0$. Hence the disk
  margin satisfies $\alpha_{max} \le \alpha_0$.  Second, it is shown
  that the closed-loop is stable and well-posed for all perturbations
  $f \in D( \alpha_0, \sigma)$.  It follows from these two steps that
  $\alpha_{max} = \alpha_0$.

  For the first step, let $\omega_0$ be the frequency where
  $S+\frac{\sigma-1}{2}$ achieves its peak gain. Define the perturbation
  $\delta_0 := \left( S(j\omega_0) +\frac{\sigma-1}{2} \right)^{-1} \in
  \C$. By construction
  $\left( S(j\omega_0) + \frac{\sigma-1}{2} \right) \delta_0 = 1$ and
  hence, by Lemma~\ref{lem:stabtrans}, the corresponding $f_0$ places
  a closed-loop pole at $s=j\omega_0$.  Moreover,
  $|\delta_0|:=\alpha_0$ and hence the corresponding $f_0$ is on the
  boundary of $D(\alpha,\sigma)$. One technical detail arises if
  $L(j\omega_0)=0$.  In this case the boundary perturbation
  $\delta_0 = \frac{2}{1+\sigma}$ yields $f_0=\infty$.  This corresponds to
  the trivial case where $D(\alpha_{max},\sigma)$ is a half-space and the
  closed-loop retains stability for any perturbation in this half
  space.

  Next show the closed-loop is stable and well-posed for all
  perturbations $f\in D(\alpha_0,\sigma)$. Each such perturbation can be
  expressed as $f=\frac{2+ (1-\sigma)\delta }{2- (1+\sigma)\delta}$ for some
  $|\delta| < \alpha_0$.  The bound $|\delta|<\alpha_0$ implies that
  $\left( S(j\omega) + \frac{\sigma-1}{2} \right) \delta \ne 1$ for all
  $\omega$.  It follows, again by Lemma~\ref{lem:stabtrans}, that the
  closed-loop has no poles on the imaginary axis for any
  $f \in D(\alpha_0,\sigma)$. Hence the closed-loop is stable for all
  $f \in D(\alpha_0,\sigma)$ because the poles for the nominal system are
  in the LHP and they do not cross the imaginary axis into the
  RHP. This can be formalized with a homotopy argument and proof by
  contradiction. Specifically, suppose the closed-loop has a pole in
  the RHP for some $f_0\in D(\alpha_0,\sigma)$. Consider the following
  equation parameterized by $0\le \tau \le 1$:
  \begin{align}
    \label{eq:HomotopyEq}
    0 & =\det\left( s I - \left( A- f(\tau) BC \right) \right) 
        \,\,\, \mbox{ where } f(\tau):= 1 + \tau(f_0-1).
  \end{align}
  For each value of $\tau$ this is a polynomial in $s$ whose roots
  correspond to the poles of the closed-loop with perturbation
  $f(\tau)$.  For $\tau=0$, this corresponds the nominal feedback
  system ($f=1$) and all roots are in the LHP by assumption.  For
  $\tau=1$, this corresponds to the perturbed feedback system $f_0$
  and there is a root in the RHP by assumption. Note that $f(\tau)$
  remains in the disk $D(\alpha_0,\sigma)$ for all $0\le \tau \le 1$.  The
  roots of a polynomial equation are continuous functions of the
  coefficients.  Hence there must be some $\tau \in [0,1]$ for which
  Equation~\ref{eq:HomotopyEq} has a root on the imaginary axis. This
  implies that the closed-loop with perturbation
  $f(\tau) \in D(\alpha_0,\sigma)$ has a pole on the imaginary
  axis. However, it has been shown that no perturbation can cause the
  closed-loop to have roots on the imaginary axis.  Thus the original
  assumption that $f_0$ causes a RHP root is false. In other words,
  the poles of the closed-loop must remain in the LHP for all
  perturbations in $D(\alpha,\sigma)$.
\end{proof}

\section{Linear Time Invariant (LTI) Perturbations}

The main disk margin result (Theorem~\ref{thm:edm}) provides a
construction for a destabilizing perturbation $f_0$. This perturbation
is a complex number with simultaneous gain and phase variation.  The
perturbation can be equivalently represented as an LTI system with
real coefficients. This equivalence is based on the following
technical lemma.

\begin{lemma}
  \label{lem:LTIinterp}
  Let a finite frequency $\omega_0>0$ and a complex number
  $\delta_0\in \C$ be given. There exists a stable, LTI system
  $\hat\delta_0$  such that $\hat\delta_0(j\omega_0)=\delta_0$ and
  $\| \hat\delta_0 \|_\infty \le |\delta_0|$.
\end{lemma}
\begin{proof}
  The basic idea is that if $\beta>0$ then
  $H(s):=\frac{s-\beta}{s+\beta}$ is stable with magnitude
  $|H(j\omega)|=1$ for all $\omega$.  This is called an all-pass
  system. Moreover, the phase of $H$ goes from $180^o$ down to $0^o$
  with increasing frequency. Similarly, $-H(s)$ is stable, all-pass
  and has phase that goes from $360^o$ up to $180^o$.  Thus a transfer
  function of the form $\pm c \frac{s-\beta}{s+\beta}$ where $c>0$ can
  achieve any desired magnitude and phase at a given frequency. The
  remainder of the proof provides details for the construction.

  If $\delta_0\in \R$ then simply select the (constant) system
  $\hat\delta_0:=\delta_0$.  Consider the alternative where
  $Im\{ \delta_0 \}\ne 0$.  In this case, $\delta_0 = \pm c e^{j\phi}$
  for some $c > 0$ and $\phi \in (0,\pi)$. Specifically, if
  $Im\{\delta_0\}>0$ then $c e^{j\phi}$ is the polar form for
  $\delta_0$.  If $Im\{\delta_0\}<0$ then it has phase
  $\angle \delta_0 \in (-\pi,0)$.  Hence $\angle \delta_0 = \phi-\pi$
  for some $\phi \in (0,\pi)$ and $\delta_0$ has the polar form
  $c e^{j(\phi-\pi)} = -c e^{j\phi}$.  

  Next, note that for $\beta>0$ the phase of
  $H(s)=\frac{s-\beta}{s+\beta}$ is given by:
  \begin{align*}
    \angle H(j\omega) 
         & = \angle ( j\omega - \beta ) - \angle ( j\omega + \beta ) \\
         & = \left[ \frac{\pi}{2} 
             + \tan^{-1} \left(\frac{\beta}{\omega} \right) \right]
              -\left[ \frac{\pi}{2}
             - \tan^{-1} \left( \frac{\beta}{\omega} \right) \right] 
          = 2 \tan^{-1} \left( \frac{\beta}{\omega} \right).
  \end{align*}
  As mentioned above, the phase of $H$ goes from $\pi$ rads down to
  $0$ as the frequency increases.  Thus $\beta$ can be selected to
  achieve the phase $\phi \in (0,\pi)$ at the specified frequency
  $\omega_0$.  Select $\beta = \omega_0 \tan\left( \phi/2 \right)$
  so that $H(j\omega_0) = e^{j\phi}$.  Finally, define
  $\hat\delta_0(s) := \pm c \frac{s-\beta}{s+\beta}$ with this $\beta$ and
  the appropriate sign for $\pm c$. Then $\hat\delta_0$ is stable with
  $\hat\delta_0(j\omega_0)=\delta_0$ and $\| \hat\delta_0 \|_\infty =|\delta_0|$.
\end{proof}

This technical lemma can be applied to obtain LTI destabilizing
perturbation from the disk margin analysis. Let $f_0$ denote a
destabilizing complex perturbation in $D(\alpha_{max},\sigma)$ with critical
frequency $\omega_0$.  This destabilizing perturbation is constructed
from a corresponding $\delta_0\in \C$ with $|\delta_0|=\alpha_{max}$.
By Lemma~\ref{lem:LTIinterp}, if $\omega_0$ is finite and nonzero then
there is a stable LTI system $\hat \delta_0$ such that
$\hat \delta_0(j\omega_0) = \delta_0$.  If $\omega_0=0$ or $\infty$
then $\delta_0$ will be real and a constant system can be selected,
i.e.  $\hat \delta_0 =\delta_0$. In either case the dynamic
perturbation $\hat{\delta}_0$ can be chosen as a constant or
first-order.  In addition, the dynamic perturbation has norm no larger
than the given uncertainty, i.e.
$\|\hat \delta_0\|_\infty \le |\delta_0| = \alpha_{max}$.  Finally,
define the following LTI perturbation:
\begin{align}
  \label{eq:fhat0}
  \hat{f}_0 = \frac{2+(1-\sigma)\hat{\delta}_0}{2-(1+\sigma)\hat{\delta}_0}.
\end{align}
This perturbation $\hat{f}_0$ is stable and on the boundary of
$D(\alpha_{max},\sigma)$ for all frequencies. The system $\hat \delta_0$
has at most one state and a minimal realization of $\hat{f}_0$ will
also have at most one state. Moreover, $\hat{f}_0(j\omega_0) = f_0$
and hence $\hat{f}_0(j\omega_0)$ causes the closed-loop to be unstable
with a pole at $s=j\omega_0$.  The LTI perturbation $\hat f_0$ can
be used within higher fidelity nonlinear simulations to gain further
insight.

\begin{ex}
  \label{ex:LTIpert}
  The symmetric disk margin was computed in Example~\ref{ex:edm} for
  the loop $L(s)=\frac{25}{s^3+10s^2+10s+10}$.  The disk margin is
  $\alpha_{max} = 0.46$ with critical frequency $\omega_0=1.94$
  rad/sec. In addition, the destabilizing perturbation
  $f_0 = 1.128-0.483j$ was constructed from
  $\delta_0 = 0.212 - 0.406j$.  The complex number $\delta_0$ has
  magnitude $0.458$ and phase $-1.089$rads.  Hence it can be expressed
  as $\delta_0 = -c e^{j\phi}$ with $c=0.458$ and $\phi = 2.052$rads.
  Select $\beta = \omega_0 \tan\left( \phi/2 \right) = 3.226$.  Based
  on the proof for Lemma~\ref{lem:LTIinterp}, the first order system
  $\hat\delta_0(s):= -0.458\, \frac{s-3.226}{s+3.226}$ is stable with
  $\hat\delta_0(j\omega_0)=\delta_0$ and
  $\| \hat\delta_0 \|_\infty =\alpha_{max}$. Equation~\ref{eq:fhat0}
  with $\sigma=0$ yields the LTI perturbation
  $\hat{f}_0 = \frac{0.627 s + 3.226}{s+2.2024}$.  It can be verified
  that $\hat f_0(j\omega_0)=f_0$ and hence the
  perturbed closed-loop sensitivity $S:=\frac{1}{1+ \hat f_0 L}$
  is unstable with a pole on the imaginary axis at $s=j\omega_0$.

\end{ex}

\newpage


\end{document}